\documentclass[12pt]{article}
\usepackage{url}
\usepackage{amsfonts}
\usepackage{amsmath,amssymb,color}
\usepackage{graphicx}
\usepackage{enumitem}
\usepackage{multirow}
\usepackage{bm}
\usepackage{url}
\usepackage{amsfonts}
\usepackage{amsmath,amssymb,color}
\usepackage{graphicx}
\usepackage{enumitem}
\usepackage{multirow}
\usepackage{bm}
\usepackage{natbib}
\usepackage{verbatim}
\usepackage{float}
\usepackage[toc,page]{appendix}
\usepackage{setspace}
\usepackage{adjustbox}
\usepackage{rotating}
\usepackage{subcaption}
\usepackage{amsthm}
\usepackage{booktabs,dcolumn,caption}
\usepackage{yhmath}
\usepackage{fullpage}
\usepackage{array}

\newcolumntype{L}[1]{>{\raggedright\let\newline\\\arraybackslash\hspace{0pt}}p{#1}}
\newcolumntype{C}[1]{>{\centering\let\newline\\\arraybackslash\hspace{0pt}}p{#1}}
\newcolumntype{R}[1]{>{\raggedleft\let\newline\\\arraybackslash\hspace{0pt}}p{#1}}

\newcommand{\MC}{\boldsymbol \beta}
\newcommand{\vc}{{\beta}}

\newcommand{\MS}{\boldsymbol \alpha}
\newcommand{\vs}{{\alpha}}

\newcommand{\ssp}{{\kappa}}
\newcommand{\sssp}{{\boldsymbol\kappa}}

\newcommand{\ttt}{\boldsymbol \theta}

\newcommand{\SG}{\boldsymbol \Sigma}

\newcommand{\EE}{\mathbf 1}

\newcommand{\EEE}{\mathbb E}

\newcommand{\MM}{\mathbf M}

\newcommand{\g}{\mathbf g}

\newcommand{\HH}{\mathbf H}

\newcommand{\hh}{\mathbf h}

\newcommand{\III}{\mathcal I}

\newcommand{\rr}{\mathbf r}
\newcommand{\R}{\mathbb R}

\newcommand{\SSSS}{\mathcal S}

\newcommand{\Y}{\mathbf Y}

\newcommand{\qq}{\mbox{$\mathbf q$}}

\newcommand{\1}{\uppercase\expandafter{\romannumeral1}}
\newcommand{\2}{\uppercase\expandafter{\romannumeral2}}

\newcommand{\diag}{\text{diag}}
\newcommand{\var}{\text{var}}
\newcommand{\avar}{\text{avar}}

\newcommand{\CI}{\text{CI}}
\newcommand{\FDR}{\text{FDR}}

\newcommand{\argmin}{\operatornamewithlimits{arg\,min}}

\newtheorem{theorem}{Theorem}

\newtheorem{lemma}{Lemma}

\newtheorem{proposition}{Proposition}

\title{Efficient estimation and inference for the signed $\beta$-model in directed signed networks}
\author{Haoran Zhang$^\dag$ and Junhui Wang$^\ddag$\\ [10pt]
	$^\dag$ Department of Statistics and Data Science \\
	Southern University of Science and Technology 
	\and
	$^\ddag$Department of Statistics \\
	The Chinese University of Hong Kong
}
\date{}

\begin{document}
\maketitle

\onehalfspacing

\begin{abstract}
This paper proposes a novel  signed $\beta$-model for directed signed network, which is frequently encountered in application domains but largely neglected in literature. The proposed signed $\beta$-model decomposes a directed signed network as the difference of two unsigned networks and embeds each node with two latent factors for in-status and out-status. The presence of negative edges leads to a non-concave log-likelihood, and a one-step estimation algorithm is developed to facilitate parameter estimation, which is efficient both theoretically and computationally. We also develop an inferential procedure for pairwise and multiple node comparisons under the signed $\beta$-model, which fills the void of lacking uncertainty quantification for node ranking. Theoretical results are established for the coverage probability of confidence interval, as well as the false discovery rate (FDR) control for multiple node comparison. The finite sample performance of the signed $\beta$-model is also examined through extensive numerical experiments on both synthetic and real-life networks.
\end{abstract}


\doublespacing

\section{Introduction}\label{Sec:intro}

Network data has attracted increasing attention from different scientific communities, due to its flexibility in describing various pairwise relations among multiple objects of interest. In literature, 
various network models have been developed, such as the Erd\"os-R\'enyi model \citep{Erdos1960evol}, the stochastic block model \citep{holland1983stochastic,Zhao2012cons}, the $\beta$-model \citep{chatterjee2011random}, the latent space model \citep{hoff2002latent}, and the network embedding model \citep{zhang2021directed}. Among them, the $\beta$-model is one of the most popular models \citep{rinaldo2013maximum, karwa2016inference, graham2017econometric, chen2021analysis}, which explicitly represents each node $i$ with a numeric factor $\beta_i$ to accommodate degree heterogeneity. Yet, most existing development of the $\beta$-model focuses on undirected and unsigned networks, and it is only recently that the directed $\beta$-model \citep{yan2016asymptotics,yan2019stat} has been developed to analyze directed unsigned networks. 

In this paper, we propose a novel signed $\beta$-model for directed signed network, which is frequently encountered in various application domains but largely neglected in literature. For instances, on many social network platforms such as Facebook or Twitter, users may send likes (positive edges) or dislikes (negative edges) to other users' posts, leading to a directed signed social network. In a citation network, authors may cite papers of other authors, where the citations can be categorized as either endorsement (positive edges) or criticism (negative edges). One interesting feature of directed signed network is the so-called status theory \citep{guha2004propagation}, which essentially suggests that a directed signed edge pointing from one node to another highly depends on their relative status. The status theory follows from the intuition that nodes with higher status tend to be more influential in the network and attract more attention from nodes with lower status \citep{leskovec2010signed}. 

Motivated by the status theory, the proposed signed $\beta$-model aims at quantifying the bi-faceted roles of each node, including its in-status and out-status. It models the probability of a directed signed edge from node $i$ to node $j$ in such a way that it is determined by both the out-status factor of node $i$ and the in-status factor of node $j$. More specifically, a node with lower out-status tends to send more negative edges and less positive edges to other nodes with low in-status, whereas a node with higher in-status tends to receive more positive edges and less negative edges from other nodes with high out-status. Furthermore, the signed $\beta$-model decomposes the directed signed network as the difference of two directed unsigned networks, corresponding to the positive and negative edges, respectively. The presence of negative edges leads to a non-concave log-likelihood, casting great challenges for parameter estimation. To circumvent the difficulty, a one-step estimation algorithm is developed, in which a single update is conducted from an initial estimate obtained via estimating equation. Besides computational efficiency, asymptotic estimation efficiency of the one-step estimate is also established.  The signed $\beta$-model also admits some novel inferential procedures for pairwise and multiple node comparisons with respect to either in-status or out-status, with theoretical guarantees on the coverage probability of confidence interval, as well as the FDR control for multiple node comparison.

{\bf Contribution.} The main contribution of this paper is three-fold. First, it proposes a novel statistical model for the under-investigated directed signed network, which decomposes it as the weighted difference of two unsigned network, and embeds each node with two latent factors for in-status and out-status. Second, it develops an efficient one-step estimation algorithm to address the non-concavity of the log-likelihood induced by the negative edges, as well as an estimation procedure for the negative sparsity parameters, which guarantees a theoretically efficient estimate and overcomes the computational-statistical gap. Third, to the best of our limited knowledge, this paper is the first attempt to provide an inferential procedure for pairwise and multiple node comparisons in directed signed networks, which fills the void of lacking uncertainty quantification for node ranking.  

{\bf Related works.} In addition to the directed $\beta$-model, there have been some recent works on directed unsigned networks, including the stochastic co-block model \citep{rohe2016co},  the sparse random graph model \citep{stein2021sparse}, and the network embedding model \citep{zhang2021directed}. All these models represent each node with two sets of latent factors, but largely rely on the nature of binary networks and cannot be directly extended to accommodate negative edges. Moreover, there have also been some other works  in machine learning literature on community detection in undirected signed networks, such as \cite{chiang2012scalable}, \cite{chiang2014prediction}, \cite{cucuringu2019sponge} and \cite{cucuringu2021regularized}. Most of these works focus on the balance theory for undirected signed network \citep{heider1946attitudes, cartwright1956structural}, which is substantially different from the status theory induced by the directed edges. 

The proposed model is also related to the Rasch model \citep{haberman1977maximum, chen2021note} in item response theory, which also relies on the exponential family assumption and thus cannot be applied to directed signed network. The Bradley-Terry model \citep{chen2019spectral, gao2021uncertainty, han2020asymptotic, chen2022optimal} has also been widely used for ranking problems, where the pairwise comparison is determined by the latent scores assigned to each item in the comparison. 
Yet, the Bradley-Terry model model is particularly designed for the ``skew-symmetric'' network, and it remains unclear how to extend the latent scores to incorporate both in-status and out-status in directed signed network.

{\bf Organization of the paper.} The rest of the paper is organized as follows. Section 2 presents the proposed signed $\beta$-model for directed signed network as well as a one-step estimation algorithm. Section 3 establishes the uniform estimation consistency and asymptotic normality of the one-step estimate. Section 4 presents the inferential procedures for pairwise and multiple node comparisons, as well as their theoretical guarantees. Section 5 conducts numerical experiments on synthetic and real-life networks to examine the finite sample performance of the proposed model. 
Section 6 concludes the paper with a brief discussion, and technical proofs and necessary lemmas are provided in the Appendix.

Throughout this paper, we use $c$ to denote a generic positive constant whose value may vary according to context. For two nonnegative sequences $a_n$ and $b_n$, $a_n\lesssim b_n$ means there exists a positive constant $c$ such that $a_n\leq cb_n$ when $n$ is sufficiently large. For a vector $\hh$, let $\hh_{[1:d]}$ denote its first $d$ entries; for a matrix $\HH$, let $\HH_{[1:d,1:d]}$ denote its upper left $d\times d$ block.

\section{Proposed method}

Suppose a directed signed network $\mathcal G$ is observed, with $n$ nodes labeled by $[n]=\{1,...,n\}$ and an adjacency matrix $\Y = (y_{ij})_{n\times n}$ with $y_{ij} \in \{-1,0,1\}$. Here, $y_{ij} = 1$ if there is a positive edge from node $i$ to node $j,$ $y_{ij} = -1$ if there is a negative edge from node $i$ to node $j,$ and $y_{ij} = 0$ if no edge is observed at all. Suppose no self loop is allowed, and thus $y_{ii}=0$ for all $i \in [n]$.

\subsection{Signed $\beta$-model}

The proposed signed $\beta$-model first decomposes $\mathcal G$ as the difference of two unsigned networks. Specifically, it formulates $y_{ij} = z^+_{ij} - z^{-}_{ij}$, where $z_{ij}^+$ and $z_{ij}^{-}$ are two independent Bernoulli random variables, and
$$
\Pr(z^+_{ij}=1) = \frac{e^{\vs_i+\vc_j}}{1+e^{\vs_i+\vc_j}},~~\mbox{and}~~\Pr(z^{-}_{ij}=1) = \frac{\ssp_i}{1+e^{\vs_i+\vc_j}}.
$$ 
Here, $\vs_i+\vc_j$ measures the relative status between nodes $i$ and $j$, and $\ssp_i \in\{\ssp_{00},\ssp_{01}\}$ with $0<\ssp_{00}<\ssp_{01}<1$ quantifies two different patterns of sending negative edges.

It is clear that as $\vs_i+\vc_j$ increases, node $i$ is more likely to send a positive edge and less likely to send a negative edge to node $j$. The probability mass function of $y_{ij}$ can be specified as
\begin{equation}
	\label{eq:model}
p(y \mid \vs_i+\vc_j,\ssp_i) =
\left\{\begin{aligned}
&\frac{e^{2(\vs_i+\vc_j)}+e^{\vs_i+\vc_j}(1-\ssp_i)}{(1+e^{\vs_i+\vc_j})^2}, ~~&& \mbox{if}~y= 1;\\
&\frac{e^{\vs_i+\vc_j}(1+\ssp_i)+ (1-\ssp_i)}{(1+e^{\vs_i+\vc_j})^2}, ~~ && \mbox{if}~y=0; \\
&\frac{\ssp_i}{(1+e^{\vs_i+\vc_j})^2}, ~~ && \mbox{if}~y=-1.
\end{aligned}
\right.
\end{equation}
It is interesting to note that \eqref{eq:model} accommodates the status theory \citep{guha2004propagation, leskovec2010signed} for directed signed network, where $\vc_j$ represents the in-status for node $j$ and $\vs_i$ represents the out-status for node $i$. It implies that a node with higher in-status tends to receive more positive edges, and a node with higher out-status tends to send more positive edges.

The signed $\beta$-model is flexible and includes the standard $\beta$-model \citep{chatterjee2011random, graham2017econometric} and the directed $\beta$-model \citep{yan2016asymptotics} as its special cases. Particularly, the signed $\beta$-model reduces to the directed $\beta$-model if all $\ssp_i$'s are set as 0, and the standard $\beta$-model if we further set $\beta_i = \alpha_i$. More interestingly, if we set $\ssp_i = 1$, the signed $\beta$-model reduce to two separate $\beta$-models, one for $z^+$ and the other for $z^-$,  except that $\EEE z^+$ increases as $\vs_i+\vc_j$ increases, while $\EEE z^-$ decreases as $\vs_i+\vc_j$ increases. However, as negative edges are often much less frequently observed than positive edges in signed networks \citep{tang2016survey}, it is more appropriate to employ small $\ssp_i$ in the signed $\beta$-model. 
In particular, we suppose $\ssp_i$ could take two different values, $\ssp_i\in\{\ssp_{00},\ssp_{01}\}$, to characterize two different patterns of sending negative edges, where $\ssp_{00} < \ssp_{01}$ and both of them may decay with $n$ to accommodate sparse networks. For example, we may employ an extremely small $\ssp_{00}$ for nodes who rarely send negative edges, while $\kappa_{01}$ could be estimated from data for those nodes who occasionally send negative edges. An estimation procedure determining the class of each $\ssp_i$ and the value of $\kappa_{01}$ is provided in the supplement.
The signed $\beta$-model is also closely related with the ordinal regression model \citep{hoff2021additive} when $y_{ij}$ is regarded as ordinal response.


Note that the parameters are not identifiable in \eqref{eq:model}, as one can add a constant to $\alpha_i$ and subtract it from $\beta_j$ without affecting the distribution of $y_{ij}$. We thus set $\vc_n = 0$ for identifiability, and  denote $\ttt = (\MS ^\top, \MC^\top)^\top\in\R^{2n-1}$ as the unknown parameters to be estimated, with $\MS = (\vs_1,...,\vs_n)^\top$ and $\MC = (\vc_1,...,\vc_{n-1})^\top$. We also denote $\sssp = (\ssp_1,...,\ssp_n)^\top$ as the sparsity parameters for negative edges. The presence of negative edges in the directed signed network casts new challenges to the analysis of the signed $\beta$-model. 
Specifically, define $l_{ij}(\vs_i+\vc_j;\ssp_i) = \log p(y_{ij}\mid\vs_i+\vc_j;\ssp_i)$ as the log-likelihood for edge $y_{ij}$. 
Then, given that $y_{ij}$'s are mutually independent, the log-likelihood function of $\mathcal G$ takes the form 
\begin{equation}\label{eq:likelihood}
l(\ttt;\sssp) = \sum_{i,j=1,i\neq j}^n l_{ij}(\vs_i+\vc_j;\ssp_i).
\end{equation}
As the positive value of $\ssp_i$ leads to a non-concave $l_{ij}$ with respect to $\vs_i$ and $\vc_j$, which further leads a non-concave log-likelihood function with respect to $\ttt$, the standard maximum likelihood estimation as in \cite{yan2016asymptotics} is no longer feasible. 
To facilitate parameter estimation, we develop an efficient one-step estimation algorithm, which does not require global optimum but still achieves asymptotical efficiency. In sharp contrast to the asymptotic analysis in \cite{yan2016asymptotics}, the signed $\beta$-model is generally not a member of exponential family, and thus it becomes substantially more challenging to quantify the asymptotic behavior of the one-step estimate.


\subsection{One-step estimation}\label{subsec:estimate}

To circumvent the non-concavity issue of $l(\ttt;\sssp)$ in \eqref{eq:likelihood}, the proposed one-step estimation algorithm performs  a single update  from an initial estimate of $\ttt$ obtained from the estimating equation approach. We assume known $\sssp$ and conduct the asymptotic analysis in the sequel, while an estimation procedure for $\sssp$ and its asymptotic properties are deferred to the supplement.

First, it follows from \eqref{eq:model} that 
$$
\EEE [y_{ij}] = \frac{e^{\vs_i+\vc_j} -\ssp_i}{1+e^{\vs_i+\vc_j}}.
$$ 
Denote $F(\ttt;\sssp) = (F_1(\ttt;\sssp),...,F_{2n-1}(\ttt;\sssp))^\top$, with
\begin{align*}
F_i(\ttt;\sssp) & = \sum_{k=1,k\neq i}^n y_{ik} - \frac{e^{\vs_i+\vc_k} -\ssp_i}{1+e^{\vs_i+\vc_k}},~\text{for}~i \in [n], \\
F_{n+j}(\ttt;\sssp) & = \sum_{k=1,k\neq j}^n y_{kj} - \frac{e^{\vs_k+\vc_j} -\ssp_k}{1+e^{\vs_k+\vc_j}},~\text{for}~j \in [n-1].
\end{align*}
Then the initial estimate of  $\ttt$ can be obtained by solving the following estimating equations,
\begin{equation}\label{eq:est1}
F(\ttt;\sssp) = \bf0.
\end{equation}
As will be shown in Theorem~\ref{thm:initial}, \eqref{eq:est1} has a unique solution, denoted as $\check\ttt =(\check \MS^\top, \check \MC^\top)^\top$. We remark that $\check\ttt$ is derived in the same way as in \cite{yan2016asymptotics}, which can be further refined via a one-step estimation algorithm.

Let $\check\vc_n = 0$, and define
\begin{equation}\label{eq:u def}
\begin{aligned}
\check u_i =& -\frac{\partial^2 l(\check\ttt;\sssp)}{\partial\vs_i^2} = -\sum_{k=1,k\neq i}^n l''_{ik}(\check\vs_i+\check\vc_k;\ssp_i),~\text{for}~i \in [n],\\
\check u_{n+j} =& -\frac{\partial^2 l(\check\ttt;\sssp)}{\partial\vc_j^2} = -\sum_{k=1,k\neq j}^n l''_{kj}(\check\vs_k+\check\vc_j;\ssp_k),~\text{for}~j \in [n-1],\\
\end{aligned}
\end{equation}
and let $\check u_{2n} = \sum_{i=1}^n \check u_i - \sum_{j=1}^{n-1} \check u_{n+j}$. As will be shown in the proof of Theorem~\ref{thm:normality}, the inverse Fisher information matrix $\left[-\partial^2 l(\check\ttt;\sssp)/\partial\ttt^2\right]^{-1}$ can be approximated by
\begin{equation}\label{eq:H def}
\check\HH  = 
\left(\begin{aligned}
  \check\HH_{11} ~&~ \check\HH_{12} \\ 
  \check\HH_{12}^\top ~&~ \check\HH_{22}
\end{aligned}\right),
\end{equation}
where $\check\HH_{11} = \diag(\check u_1^{-1},...,\check u_n^{-1}) + \check u_{2n}^{-1}\EE_{n}\EE_{n}^\top$, $\check\HH_{22} = \diag(\check u_{n+1}^{-1},...,\check u_{2n-1}^{-1}) + \check u_{2n}^{-1}\EE_{n-1}\EE_{n-1}^\top$, and $\check\HH_{12} = -\check u_{2n}^{-1}\EE_{n}\EE_{n-1}^\top$.
Here $\EE_{2n-1}$ denotes a vector with all ones. 
Then the one-step estimate is given as
\begin{equation}\label{eq:est2}
\widehat\ttt  = \check\ttt + \check\HH \left(\frac{\partial l(\check\ttt;\sssp)}{\partial\ttt}\right),
\end{equation}
which is equivalent to 
$$
\begin{aligned}
\widehat\vs_i &= \check\vs_i + \check u_{i}^{-1} \frac{\partial l(\check\ttt;\sssp)}{\partial\vs_i} + \check u_{2n}^{-1}\sum_{k=1}^{n}\frac{\partial l(\check\ttt;\sssp)}{\partial\vs_k} - \check u_{2n}^{-1}\sum_{l=1}^{n-1}\frac{\partial l(\check\ttt;\sssp)}{\partial\vc_l},~\text{for}~i \in [n],\\
\widehat\vc_j &= \check\vc_j + \check u_{n+j}^{-1} \frac{\partial l(\check\ttt;\sssp)}{\partial\vc_j} - \check u_{2n}^{-1}\sum_{k=1}^{n}\frac{\partial l(\check\ttt;\sssp)}{\partial\vs_k} + \check u_{2n}^{-1}\sum_{l=1}^{n-1}\frac{\partial l(\check\ttt;\sssp)}{\partial\vc_l},~\text{for}~j \in [n-1].
\end{aligned}
$$ 
The final estimate is denoted as $\widehat\ttt = (\widehat\MS^\top, \widehat\MC^\top)^\top =  (\widehat\vs_1,...,\widehat\vs_n,\widehat\vc_1,...,\widehat\vc_{n-1})$ and $\widehat\vc_n = 0$. It is worthy pointing out that the one-step estimation in \eqref{eq:est2} needs not to calculate the inverse Hessian matrix as standard Newton-Raphson update, and thus is computationally more efficient. More importantly, this one-step estimation also attains asymptotic estimation efficiency without assuming the intractable global optimum, as will be shown in Theorem~\ref{thm:normality}.

\section{Asymptotic theory}\label{sec:theory}

This section establishes the uniform consistency and asymptotic normality of the one-step estimate $\widehat \ttt$ in Section \ref{subsec:estimate}. Let $\ttt^* = (\vs_1^*,...\vs_{n}^*,\vc_1^*,...,\vc_{n-1}^*)^\top$ denote the true parameters, and $\| \ttt^* \|_{\infty} = \max\{|\vs_1^*|,...,|\vs_{n}^*|,|\vc_1^*|,...,|\vc_{n-1}^*|\}$.

For $i,j \in [n]$, let
 \begin{equation}\label{eq:uvw def}
 \begin{aligned}
 	&u_i = \EEE \left[-\frac{\partial^2 l(\ttt^*;\sssp)}{\partial \vs_i^2} \right], &&v_i = \sum_{k=1,k\neq i}^n \frac{(1+\ssp_i)e^{\vs_i^*+\vc_k^*}}{(1+e^{\vs_i^*+\vc_k^*})^2},&&w_i = \sum_{k=1,k\neq i}^n\var(y_{ik}),\\
 	&u_{n+j} = \EEE \left[-\frac{\partial^2 l(\ttt^*;\sssp)}{\partial \vc_j^2} \right],&&v_{n+j} = \sum_{k=1,k\neq j}^n \frac{(1+\ssp_k)e^{\vs_k^*+\vc_j^*}}{(1+e^{\vs_k^*+\vc_j^*})^2},&&w_{n+j}=\sum_{k=1,k\neq j}^n\var(y_{kj}),
 \end{aligned}
 \end{equation} 
 where $u_{2n} = \EEE \left[-\partial^2 l(\ttt^*;\sssp)/\partial \vc_n^2 \right]$ is defined by $u_{2n} = \sum_{i=1}^n u_i - \sum_{j=1}^{n-1} u_{n+j}$. We define two matrices as following, which will be shown as the asymptotic covariance matrices for $\check\ttt$ and $\widehat\ttt$ in Theorems~\ref{thm:initial} and \ref{thm:normality}, $$
\SG = \left(\begin{aligned}
  \SG_{11} ~&~ \SG_{12} \\ 
  \SG_{12}^\top ~&~ \SG_{22}
\end{aligned}\right)
~~~~\text{and}~~~~
\HH = \left(\begin{aligned}
  \HH_{11} ~&~ \HH_{12} \\ 
  \HH_{12}^\top ~&~ \HH_{22}
\end{aligned}\right),
$$ where $\SG_{12} = -w_{2n} v_{2n}^{-2}\EE_{n}\EE_{n-1}^\top$, $\HH_{12} = -u_{2n}^{-1}\EE_{n}\EE_{n-1}^\top$, and $$
\begin{aligned}
&\SG_{11} = \diag(w_1 v_1^{-2},...,w_{n} v_n^{-2}) + w_{2n} v_{2n}^{-2}\EE_{n}\EE_{n}^\top,\\
&\SG_{22} = \diag(w_{n+1} v_{n+1}^{-2},...,w_{2n-1} v_{2n-1}^{-2}) + w_{2n} v_{2n}^{-2}\EE_{n-1}\EE_{n-1}^\top,\\
&\HH_{11} = \diag(u_1^{-1},...,u_n^{-1}) + u_{2n}^{-1}\EE_{n}\EE_{n}^\top,\\
&\HH_{22} = \diag(u_{n+1}^{-1},..., u_{2n-1}^{-1}) + u_{2n}^{-1}\EE_{n-1}\EE_{n-1}^\top.
\end{aligned}
$$

We first establish the uniform consistency and asymptotic normality of the initial estimate $\check\ttt$ from the first step. 

 \begin{theorem}\label{thm:initial}
 	Suppose $\|\ttt^*\|_{\infty} \leq c\log n$ with $0<c<1/40$ and $\|\boldsymbol\ssp - \boldsymbol\ssp^*\|_{\infty} \lesssim e^{12\|\ttt^*\|_{\infty}}\log n/n$.
	Then as $n$ goes to infinity, with probability at least $1-c_1/n$ for a constant $c_1$, it holds true that \eqref{eq:est1} has a unique solution $\check\ttt$, which satisfies that
 	\begin{equation}\label{eq:consis0}
 		\|\check\ttt-\ttt^*\|_{\infty} \lesssim e^{6\|\ttt^*\|_{\infty}} \sqrt{\frac{\log n}{n}}.
 	\end{equation} 
 	Further, for any fixed $d$, $(\check\ttt-\ttt^*)_{[1:d]}$ is asymptotically multivariate normal with mean ${\bf0}$ and covariance matrix given by the upper $d\times d$ block of $\SG$.
\end{theorem}


Theorem \ref{thm:initial} shows that the initial estimate $\check \ttt$ is a fairly good estimate and converges to $\ttt^*$ at a fast rate. The asymptotic variance of $\check\theta_i$ is given as
$$
\avar(\check\theta_i)= w_i v_i^{-2}+w_{2n} v_{2n}^{-2},
$$
where the term $w_{2n} v_{2n}^{-2}$ is due to the identifiability constraint that $\vc_n = 0$. Further, if $\cal G$ is an unsigned network with $\ssp_1=...=\ssp_n=0$, then $w_i = v_i$ and $\avar(\check\theta_i)= v_i^{-1}+ v_{2n}^{-1}$, which coincides with the result in Theorem 2 of \cite{yan2016asymptotics}. More interestingly, Theorem \ref{thm:initial} holds true for sparse signed networks by allowing $\|\ttt^*\|_{\infty} \leq c\log n$, which matches up with the existing sparsity results for unsigned $\beta$-model \citep{yan2016asymptotics}.

We are now ready to establish the consistency and asymptotic normality of $\widehat\ttt$.

 \begin{theorem}\label{thm:normality}
	Under the same condition of Theorem~\ref{thm:initial}, with probability at least $1-c_2/n$ for a constant $c_2$, we have 
	\begin{equation}\label{eq:consis1}
 		\|\widehat\ttt-\ttt^*\|_{\infty} \lesssim e^{2\|\ttt^*\|_{\infty}} \sqrt{\frac{\log n}{n}}.
 	\end{equation} 
	Further, for any fixed $d$, $(\widehat\ttt-\ttt^*)_{[1:d]}$ is asymptotically multivariate normal with mean ${\bf0}$ and covariance matrix given by the upper $d\times d$ block of $\HH$. 
\end{theorem}

 
 Theorem \ref{thm:normality} shows that the one-step estimate $\widehat \ttt$ also converges to $\ttt^*$ at a fast rate, and its asymptotic variance is given as
 $$
 \avar(\widehat\theta_i) = u_i^{-1}+u_{2n}^{-1},
 $$ where the term $u_{2n}^{-1}$ is also due to the identifiability constraint that $\vc_n = 0$. It is important to remark that the one-step estimate $\widehat \ttt$ is as efficient as the global maximizer of the non-concave log-likelihood function of \eqref{eq:likelihood}, which is difficult to obtain in directed signed networks, if not impossible. Proposition~\ref{prop:var compare} shows that $\widehat \ttt$ outperforms the initial estimate $\check \ttt$ asymptotically by reducing the estimation variance, confirming the advantage of the one-step estimate. 
 
\begin{proposition}\label{prop:var compare}
Under the same conditions of Theorem~\ref{thm:initial}, we have $\avar(\check\theta_i) \geq \avar(\widehat\theta_i)$.
\end{proposition}

We should point out that the efficiency gain in Proposition 1 is non-negligible, even in the sparse regime. For example, suppose $\vc_1^* = ... = \vc_n^* = 0$, and let $\vs_i^*\to-\infty$ and $\ssp_i\to0$, and then both positive and negative edges are sparse. It can be verified that $\frac{w_i v_i^{-2}}{u_i^{-1}} \to1+ \frac{\ssp_i}{e^{\vs_i^*}}$, and thus
$\frac{\avar(\check\theta_i)}{\avar(\widehat\theta_i)} > 1$ as long as $\ssp_i = \omega(e^{\vs_i^*})$, which implies a non-negligible gain in terms of the asymptotic variance.

\section{Inference for node ranking}\label{sec:infer}

Node ranking has been an important task in network data analysis \citep{wasserman1994social}, which aims to rank nodes based on their importance or centrality. 
In literature, many ranking algorithms have been developed for node ranking in directed unsigned networks, including \cite{freeman1978centrality, latora2007measure, page1999pagerank, kleinberg1999authoritative}. These algorithms have also been extended to directed signed network, such as \cite{bonacich2004calculating, zolfaghar2010mining, shahriari2014ranking}; readers may refer to \cite{tang2016survey} for a complete literature review on node ranking in directed signed networks. Despite the rich literature, most aforementioned algorithms are intuition driven and lack of theoretical justification, not to mention developing an inferential framework to conduct uncertainty quantification for node ranking. Based on the signed $\beta$-model, this section develops some inferential procedures for pairwise and multiple node comparisons with respect to either in-status or out-status.


\subsection{Pairwise comparison}

The statistical inference for $\vs_i^* - \vs_j^*$ and $\vc_i^* - \vc_j^*$ represents the relative differences between nodes $i$ and $j$ in terms of their in-status and out-status, respectively. For $i\in[n]$ and $j\in[n-1],$ define 
\begin{equation}\label{eq:u def2}
\begin{aligned}
\widehat u_i = -\frac{\partial^2 l(\widehat\ttt;\sssp)}{\partial\vs_i^2},~~\text{and}~~ 
\widehat u_{n+j} = -\frac{\partial^2 l(\widehat\ttt;\sssp)}{\partial\vc_j^2}.
\end{aligned}
\end{equation}
Further,  for any $i\neq j\in[2n-1]$, define 
\begin{equation}\label{eq:delta def}
\widehat\delta_{ij}^{2} = \widehat u_i^{-1} + \widehat u_j^{-1},~~\text{and}~~(\delta_{ij}^*)^{2} = u_i^{-1} + (u_j^*)^{-1}.
\end{equation} 
We first establish the asymptotic normality for both $(\widehat\vs_i - \widehat\vs_j)$ and $(\widehat\vc_i - \widehat\vc_j)$.

\begin{theorem}\label{thm:ranking}
	Under the same condition of Theorem~\ref{thm:initial}, 
	for any $i\neq j\in[n],$ it holds true that 
	\begin{equation}\label{eq:rank normality}
		\begin{aligned}
			(\delta_{ij}^*)^{-1} \left[(\widehat\vs_i - \widehat\vs_j) - (\vs_i^*- \vs_j^*)\right] &\to N(0,1),\\
			(\delta_{n+i,n+j}^*)^{-1}\left[(\widehat\vc_i - \widehat\vc_j) - (\vc_i^*- \vc_j^*)\right]&\to N(0,1)
		\end{aligned}
	\end{equation}
	in distribution.
	Furthermore, we also have
	\begin{equation}\label{eq:CI}
		\begin{aligned}
			\widehat\delta_{ij}^{-1} \left[(\widehat\vs_i - \widehat\vs_j) - (\vs_i^*- \vs_j^*)\right] &\to N(0,1),\\
			\widehat\delta_{n+i,n+j}^{-1}\left[(\widehat\vc_i - \widehat\vc_j) - (\vc_i^*- \vc_j^*)\right]&\to N(0,1)
		\end{aligned}
	\end{equation}
	in distribution.
\end{theorem}

According to \eqref{eq:rank normality}, the asymptotic variance for $\widehat\vs_i - \widehat\vs_j$ is $\delta_{ij}^*$ which is defined in \eqref{eq:delta def}, suggesting that the proposed estimate for $\vs_i^*- \vs_j^*$ is oracle, in the sense that its asymptotic distribution is the same as the maximum likelihood estimate with known $\{\vs_k\}_{k\neq i,j}$ and $\{\vc_j\}_{j=1}^{n-1}$.

Furthermore, given the asymptotic normality results in Theorem \ref{thm:normality}, we can construct a confidence interval for $\vs_i^* - \vs_j^*$ as 
$$
\CI(\vs_i^*-\vs_j^*) = \left[(\widehat\vs_i - \widehat\vs_j) - Z_{\alpha/2}\widehat\delta_{ij}, (\widehat\vs_i - \widehat\vs_j) + Z_{\alpha/2}\widehat\delta_{ij}  \right],
$$ 
and a confidence interval  for $\vc_i - \vc_j$ as 
$$
\CI(\vc_i^*-\vc_j^*) = \left[(\widehat\vc_i - \widehat\vc_j) - Z_{\alpha/2}\widehat\delta_{n+i,n+j}, (\widehat\vc_i - \widehat\vc_j) + Z_{\alpha/2}\widehat\delta_{n+i,n+j}  \right],
$$ 
where $Z_{\alpha}$ denotes the $\alpha$-th upper percentile of the standard normal distribution. Furthermore, we can also test whether $\vs_i^* > \vs_j^*$ based on the indicator $1_{\{(\widehat\vs_i - \widehat\vs_j) - Z_{\alpha/2}\widehat\delta_{ij} > 0\}},$ and similarly for testing $\vc_i^* > \vc_j^*$.

\subsection{Multiple comparison}

We now focus on some particular node $i\in[n]$, and find its relative rank within a subgroup of nodes. We take out-status for illustration, and similar procedure can be developed for in-status. Particularly, let $\SSSS\subseteq[n]/\{i\}$ be the subgroup of nodes, and $K = |\SSSS|$. For each $k\in\SSSS,$ we want to test 
$$
H^{(k)}_{0}: \vs_i^*=\vs_k^*~~\text{v.s.}~~H^{(k)}_{a}:\vs_i^*\neq\vs_k^*.
$$ 

We employ the Benjamini-Hochberg procedure \citep{benjamini2001control} to control the false discovery rate for this multiple testing problem. Let $p_k$ denote the p-value for testing $H_0^{(k)},$ which takes the form 
$$
p_k = 2\left[1 - \Phi\left(\widehat\delta_{ik}^{-1}|\widehat\vs_i-\widehat\vs_k|\right)\right],
$$ 
where $\Phi(\cdot)$ is the cumulative distribution function of the standard normal distribution. We order the $K$ p-values as $p^{(1)}\leq...\leq p^{(K)}$, and define 
\begin{equation}\label{eq:cutoff}
r = \max_{1\leq l\leq K}\left\{l:p^{(l)}\leq\frac{\alpha l}{KL}\right\},
\end{equation}
which is a modification of the traditional Benjamini-Hochberg procedure \citep[Theorem 1.3,][]{benjamini2001control}. Here $\alpha>0$ is a given significance level, and $L = \sum_{l=1}^K1/l.$ 
Then, for any $k\in\SSSS,$ we reject $H_0^{(k)}$ if $p_k \leq p^{(r)}.$ If the set in \eqref{eq:cutoff} is empty, we reject all null hypotheses. 
Let $\SSSS_0 = \{k\in\SSSS: \vs_i^*=\vs_k^*\}$ be the set of true null hypotheses, then the false discovery rate for the Benjamini-Hochberg procedure takes the form
$$
\FDR = \EEE\left[\frac{\sum_{k\in\SSSS_0} 1_{\{p_k\leq p^{(r)}\}}}{\max\left\{ \sum_{k\in\SSSS} 1_{\{p_k\leq p^{(r)}\}},1 \right\}}\right] = \EEE\left[\frac{\sum_{k\in\SSSS_0}1_{\{p_k\leq\frac{\alpha r}{KL}\}}}{\max\{r,1\}}\right].
$$

\begin{theorem}\label{thm:fdr}
Under the same condition of Theorem~\ref{thm:initial}, further suppose $e^{20\|\ttt^*\|_{\infty}}K_0n^{-1/2}(\log n)^2 = o(1)$, with $K_0 = |\SSSS_0|$. Then, we have 
$$
\FDR \leq \frac{\alpha K_0}{K}\left(1+\frac{1}{L}\right) + o(1).
$$ 
\end{theorem}

Theorem \ref{thm:fdr} immediately implies that when $K_0$ is dominated by $e^{-20\|\ttt^*\|_{\infty}}n^{1/2}(\log n)^{-2},$ the false discovery rate of the Benjamini-Hochberg procedure is upper bounded by $\alpha$ asymptotically as long as $\frac{K_0}{K} \big ( 1+\frac{1}{L} \big ) \leq 1$, which is a fairly mild condition since $K_0\leq K$ and $L$ is roughly $\log K$.

\section{Numerical experiments}

This section examines the finite sample performance of the proposed one-step estimate as well as the inferential procedures, where the sparse factor $\sssp$ is estimated as described in the supplement and the estimating equation in \eqref{eq:est1} is solved by Newton's method.

\subsection{Simulation}\label{subsec:simu}

The simulated directed signed networks are generated as follows. We first generate 10 groups of nodes, where nodes in the same group have the same in-status and out-status. Specifically, let $\psi_i \in \{1,\ldots,10\}$ denote the group membership of node $i$, generated independently from a multinomial distribution with probabilities $(0.15 \times {\bf 1}_5^\top, 0.05 \times {\bf 1}_5^\top)$ to accommodate unbalanced groups. We then set $\vs_i^* = a_{\psi_i}\sim N(-0.5, 0.5),~\vc_i^* = b_{\psi_i}\sim N(0, 0.5)$ and $\vc_n^* = 0.$ The directed edges, $y_{ij},$ are then generated independently from \eqref{eq:model}. The $\ssp_i$ are randomly generated from $\{\ssp_{00},\ssp_{01}\}$ with $\Pr(\ssp_i=\ssp_{01}) = 0.8$. Various scenarios are considered, with $n \in \{200, 600, 1000\}$, $\ssp_{01} \in \{0.05, 0.1, 0.2, 0.25\}$ and $\ssp_{00} = 0.001$. 

In each scenario, the averaged estimation errors, measured by $\|\widehat\ttt-\ttt^*\|_{\infty}$ and $\|\widehat\ttt-\ttt^*\|_2^2/(2n-1)$, over 500 independent replications with their standard errors are reported in Tables 1 and 2.  The averaged coverage frequencies of the 95\% confidence interval for 100 randomly selected $\widehat\vs_i - \widehat\vs_j,$ together with their standard errors, are reported in Table 3.  In addition, we randomly select a node from the group with the largest in-status, and compare the in-status of this node with the rest nodes in this group and the first 10 nodes from each of the other  groups. The averaged false discovery proportions (FDP) and power, as well as their standard errors, are reported in Tables 4 and 5. We also report these evaluation metrics of the initial estimate $\check\ttt$ for comparison.


\begin{table}[!h]
\begin{center}
\caption{The averaged estimation errors over 500 independent replications and their standard errors in parenthesis}
\label{tab:consis}
\begin{small}
\begin{tabular}{ c|c|c|c|c|c } 
\hline
\hline
 & $n$ & $\ssp_{01} = 0.05$ & $\ssp_{01} = 0.1$ & $\ssp_{01} = 0.2$ & $\ssp_{01} = 0.25$   \\
\hline
\hline
\multirow{3}{5em}{$\|\check\ttt-\ttt^*\|_{\infty}$} & 200 & 0.6087 (0.1098) & 0.6210 (0.1179)& 0.6873 (0.1296)& 0.7483 (0.1467) \\ \cline{2-6} 
& 600 &0.3756 (0.0617)& 0.4415 (0.0697)& 0.6474 (0.0923)& 0.7528 (0.1007) \\ \cline{2-6} 
& 1000& 0.3057 (0.0483)& 0.3929 (0.0686)& 0.6352 (0.0862)& 0.7445 (0.0865) \\ \cline{2-6} 
\hline
\hline
\multirow{3}{5em}{$\|\widehat\ttt-\ttt^*\|_{\infty}$} & 200 &0.5905 (0.1003)& 0.5908 (0.1064)& 0.6196 (0.1103)& 0.6616 (0.1237) \\ \cline{2-6} 
& 600 &0.3603 (0.0577)& 0.3820 (0.0581)& 0.5070 (0.0784)& 0.5939 (0.0892) \\ \cline{2-6} 
& 1000 &0.2899 (0.0453)& 0.3133 (0.0511)& 0.4666 (0.0749)& 0.5690 (0.0759) \\ \cline{2-6} 
\hline
\hline
\end{tabular}
\end{small}
\end{center}
\end{table}

\begin{table}[!h]
\begin{center}
\caption{The averaged mean squared errors over 500 independent replications and their standard errors in parenthesis}
\label{tab:consis_check}
\begin{small}
\begin{tabular}{ c|c|c|c|c|c } 
\hline
\hline
 & $n$ & $\ssp_{01} = 0.05$ & $\ssp_{01} = 0.1$ & $\ssp_{01} = 0.2$  & $\ssp_{01} = 0.25$ \\
\hline
\hline
\multirow{3}{5em}{{\large $\frac{\|\check\ttt-\ttt^*\|_2^2}{2n-1}$}} & 200 &0.0467 (0.0296)& 0.0473 (0.0299)& 0.0508 (0.0324)& 0.0525 (0.0309) \\ \cline{2-6} 
& 600 &0.0147 (0.0103)& 0.0156 (0.0093)& 0.0180 (0.0085)& 0.0195 (0.0090) \\ \cline{2-6} 
& 1000 &0.0091 (0.0067)& 0.0093 (0.0060)& 0.0107 (0.0070)& 0.0114 (0.0066) \\ \cline{2-6} 
\hline
\hline
\multirow{3}{5em}{{\large $\frac{\|\widehat\ttt-\ttt^*\|_2^2}{2n-1}$}} & 200 &0.0449 (0.0289)& 0.0449 (0.0292)& 0.0468 (0.0298)& 0.0478 (0.0290) \\ \cline{2-6} 
& 600 &0.0141 (0.0100)& 0.0144 (0.0093)& 0.0157 (0.0082)& 0.0168 (0.0084) \\ \cline{2-6} 
& 1000 &0.0088 (0.0067)& 0.0087 (0.0061)& 0.0096 (0.0066)& 0.0100 (0.0063) \\ \cline{2-6} 
\hline
\hline
\end{tabular}
\end{small}
\end{center}
\end{table}

\begin{table}[!h]
\begin{center}
\caption{The averaged coverage frequencies of the 95\% confidence interval over 500 independent replications for 100 randomly selected pairs $(i,j),$ and their standard errors in parenthesis}
\label{tab:infer1}
\begin{small}
\begin{tabular}{ c|c|c|c|c|c } 
\hline
\hline
 & $n$ & $\ssp_{01} = 0.05$ & $\ssp_{01} = 0.1$ & $\ssp_{01} = 0.2$ & $\ssp_{01} = 0.25$  \\
\hline
\hline
\multirow{3}{5em}{Coverage Frequency ($\check\ttt$)} & 200 &0.9371 (0.0185)& 0.9341 (0.0245)& 0.9051 (0.0603)& 0.8886 (0.0817) \\ \cline{2-6} 
& 600 &0.9423 (0.0166)& 0.9206 (0.0491)& 0.8887 (0.0921)& 0.8852 (0.0944) \\ \cline{2-6} 
& 1000 &0.9431 (0.0145)& 0.9319 (0.0339)& 0.9248 (0.0439)& 0.9211 (0.0472) \\ \cline{2-6} 
\hline
\hline
\multirow{3}{5em}{Coverage Frequency ($\widehat\ttt$)} & 200 &0.9441 (0.0124)& 0.9419 (0.0155)& 0.9218 (0.0388)& 0.9043 (0.0610) \\ \cline{2-6} 
& 600 &0.9477 (0.0098)& 0.9353 (0.0237)& 0.8979 (0.0778)& 0.8895 (0.0894) \\ \cline{2-6} 
& 1000 &0.9466 (0.0101)& 0.9399 (0.0188)& 0.9262 (0.0415)& 0.9217 (0.0471) \\ \cline{2-6} 
\hline
\hline
\end{tabular}
\end{small}
\end{center}
\end{table}

\begin{table}[!h]
\begin{center}
\caption{The averaged FDP over 500 independent replications and their standard errors in parenthesis}
\label{tab:infer2}
\begin{small}
\begin{tabular}{ c|c|c|c|c|c } 
\hline
\hline
 & $n$ & $\ssp_{01} = 0.05$ & $\ssp_{01} = 0.1$ & $\ssp_{01} = 0.2$ & $\ssp_{01} = 0.25$  \\
\hline
\hline
\multirow{3}{5em}{FDP ($\check\ttt$)} & 200 &0.0013 (0.0062)& 0.0015 (0.0063)& 0.0030 (0.0094)& 0.0051 (0.0122) \\ \cline{2-6} 
& 600 &0.0027 (0.0095)& 0.0065 (0.0176)& 0.0289 (0.0539)& 0.0341 (0.0618) \\ \cline{2-6} 
& 1000 &0.0053 (0.0233)& 0.0077 (0.0301)& 0.0182 (0.0572)& 0.0240 (0.0676) \\ \cline{2-6} 
\hline
\hline
\multirow{3}{5em}{FDP ($\widehat\ttt$)} & 200 &0.0014 (0.0069)& 0.0014 (0.0060)& 0.0027 (0.0089)& 0.0042 (0.0112) \\ \cline{2-6} 
& 600 & 0.0025 (0.0089)& 0.0048 (0.0143)& 0.0225 (0.0447)& 0.0296 (0.0570) \\ \cline{2-6} 
& 1000 &0.0050 (0.0229)& 0.0065 (0.0265)& 0.0154 (0.0512)& 0.0227 (0.0652) \\ \cline{2-6} 
\hline
\hline
\end{tabular}
\end{small}
\end{center}
\end{table}

\begin{table}[!h]
\begin{center}
\caption{The averaged power over 500 independent replications and their standard errors in parenthesis}
\label{tab:infer3}
\begin{small}
\begin{tabular}{ c|c|c|c|c|c } 
\hline
\hline
 & $n$ & $\ssp_{01} = 0.05$ & $\ssp_{01} = 0.1$ & $\ssp_{01} = 0.2$ & $\ssp_{01} = 0.25$  \\
\hline
\hline
\multirow{3}{5em}{Power $(\check\ttt)$} & 200 &0.6610 (0.1068)& 0.6440 (0.1220)& 0.6553 (0.1355)& 0.6652 (0.1410) \\ \cline{2-6} 
& 600 &0.7871 (0.0712)& 0.7935 (0.0767)& 0.8029 (0.0920)& 0.8042 (0.0937) \\ \cline{2-6} 
& 1000 &0.8370 (0.0608)& 0.8388 (0.0609)& 0.8432 (0.0634)& 0.8468 (0.0624) \\ \cline{2-6} 
\hline
\hline
\multirow{3}{5em}{Power $(\widehat\ttt)$} & 200 &0.6531 (0.1087)& 0.6428 (0.1185)& 0.6558 (0.1303)& 0.6670 (0.1346)  \\ \cline{2-6} 
& 600 &0.7863 (0.0712)& 0.7935 (0.0753)& 0.8041 (0.0872)& 0.8071 (0.0898) \\ \cline{2-6} 
& 1000 &0.8375 (0.0605)& 0.8395 (0.0605)& 0.8438 (0.0625)& 0.8494 (0.0609) \\ \cline{2-6} 
\hline
\hline
\end{tabular}
\end{small}
\end{center}
\end{table}

It is evident from Tables 1 and 2 that the estimation errors for both $\check\ttt$ and $\widehat\ttt$ decrease as $n$ increases, which validates the asymptotic estimation consistency in Theorems \ref{thm:initial} and \ref{thm:normality}. Also, the performance of $\widehat\ttt$ is consistently better than that of $\check\ttt$ in all scenarios, which is expected according to Theorem \ref{thm:normality}. It is also interesting to note that the estimation errors of $\widehat\ttt$ are fairly robust against the sparsity level of the negative edges. In Table 3, it is clear that the coverage frequencies for the 95\% confidence interval are close to the nominal level as $n$ grows, which supports the asymptotic normality results in Theorem \ref{thm:ranking}. As for multiple testing, as shown in Tables 4 and 5, the false discovery proportion is way below 0.05 in all scenarios and the power increases as $n$ grows, justifying the false discovery rate control in Theorem \ref{thm:fdr}. 
Similarly, the performance of $\widehat\ttt$ is better than $\check\ttt$.

\subsection{Signed citation network}

We now apply the proposed method to analyze a real-life signed citation network \citep{kumar2016structure}, which is collected based on citations within the natural language processing (NLP) community during 1975-2013. Particularly, each author is represented as a node, which sends edges to other nodes by citing their papers. According to the citation sentiments, the citations can be categorized into ``endorsement'', ``criticism'' and ``neural''. Both ``endorsement'' and ``neural'' are treated as positive edges, while ``criticism'' is treated as negative edge. We remove all authors who send or receive less than 5 edges, as well as about 80 authors who send or receive more negative edges than positive edges, leading to spuriously high out-status or low in-status. This pre-processing step leads to a directed network with 1849 nodes, 56517 positive edges and 3726 negative edges. We set $\ssp_i = \ssp_{00} = 0.001$ for those nodes who do not send negative edges, and $\ssp_i = \ssp_{01}$ for the rest, with $\ssp_{01}$ estimated as described in the supplement.


The left panel of Figure \ref{fig:barplot citation} shows the in-degrees for the 10 nodes with highest estimated in-status.  In particular, each grey bar above the x-axis
shows the positive in-degree $\sum_{j=1,j\neq i}^ny_{ji}1_{\{y_{ji}>0\}},$ while the red bar below the x-axis shows the negative in-degree $\sum_{j=1,j\neq i}^n|y_{ji}|1_{\{y_{ji}<0\}}.$ The black line further shows the estimated $\widehat\vc_i$ for the 10 nodes after rescaling.
The right panel of Figure~\ref{fig:barplot citation} shows the out-degrees for the 10 nodes with highest estimated out-status. It can be seen that the estimated in-status is highly positively correlated with the positive in-degree, reflecting the attractiveness or popularity of the node; whereas the estimated out-status is highly positively correlated with the positive out-degree, reflecting the social status of the node to some extent. It is also interesting to note that in the right panel, the 6th node actually sends out less positive edges than the 7th node, but it possesses a higher out-status as it also sends out less negative edges.

To further scrutinize the results in Figure \ref{fig:barplot citation}, we note that the 10 nodes with highest estimated in-status are Christopher Manning, Daniel Klein, Michael Collins, Salim Roukos, Franz Josef Och, Fernando Pereira, Daniel Marcu, Philipp Koehn, Vincent J. Della Pietra and Eugene Charniak. Clearly, they are all highly cited researchers in the NLP community, with google scholar citation counts ranging from 25000 to 189000. We also note that the 10 nodes with highest estimated out-status are Joakim Nivre, Noah Smith, Mirella Lapata, Timothy Baldwin, Christopher Manning, Chris Callison-Burch, Junichi Tsujii, Eduard Hovy, Dan Roth and Regina Barzilay. These researchers are very active and productive, with number of publications ranging from 254 to 732, according to google scholar. It is also interesting to note that Christopher Manning is among both lists, who is well regarded as an influential and productive NLP expert.

\begin{figure}[!htb]
	\centering
	\includegraphics[scale = 0.4]{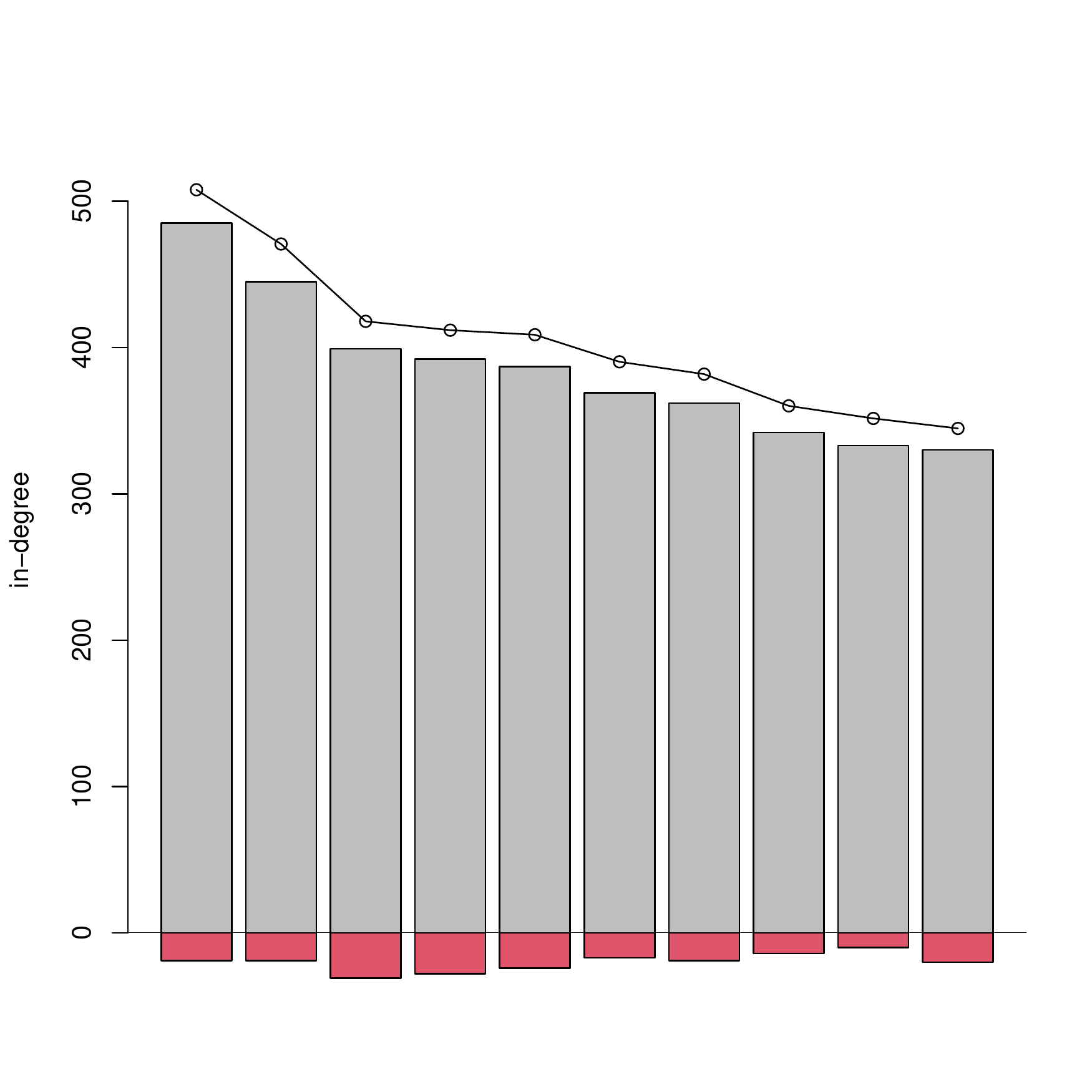}
	\includegraphics[scale = 0.4]{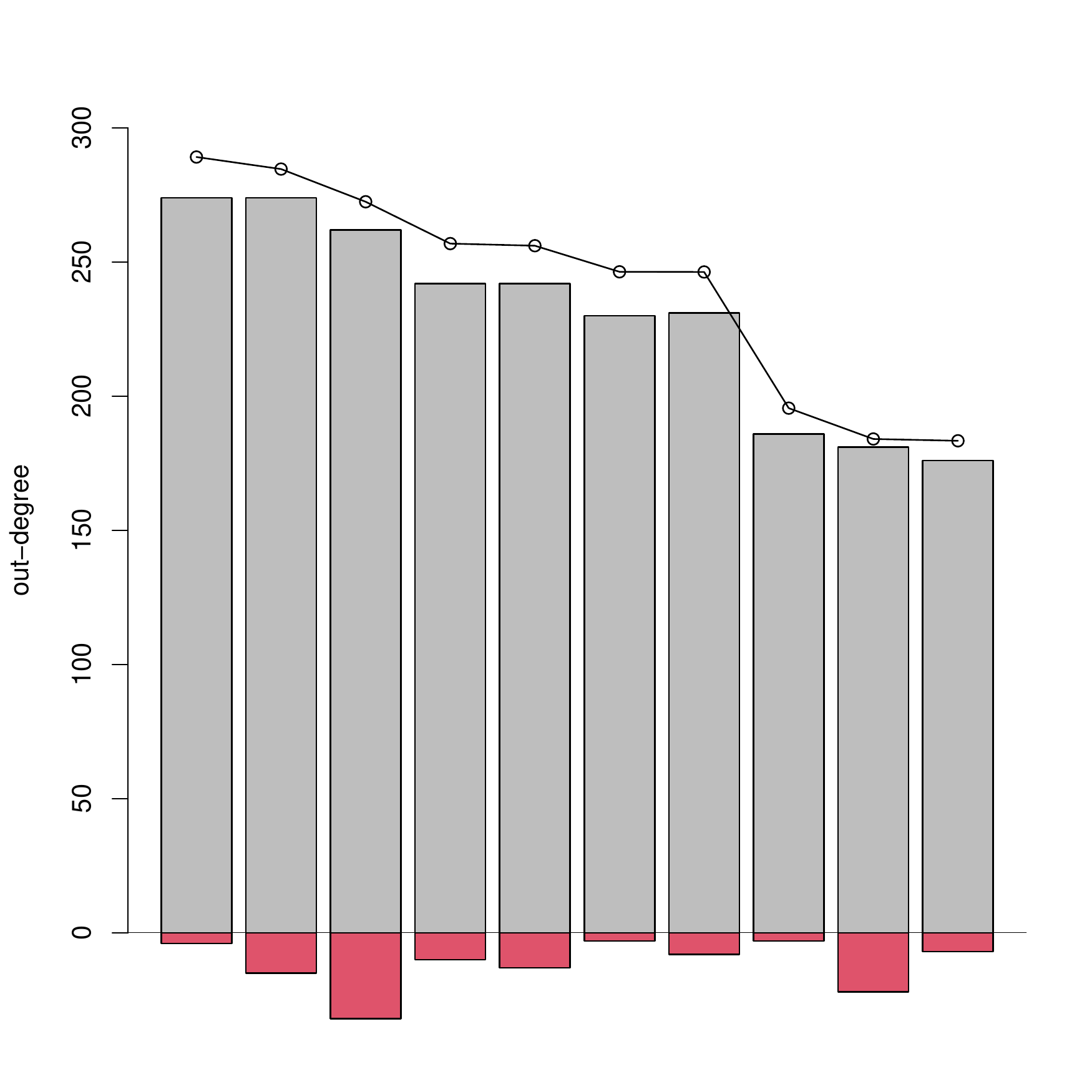}
	\caption{The left panel shows the positive in-degree (above x-axis) and negative in-degree (below x-axis) for the 10 nodes with largest $\widehat\vc_i.$ The black line shows the corresponding $\widehat\vc_i$ after rescaling. The right panel shows the positive out-degree (above x-axis) and negative out-degree (below x-axis) for the 10 nodes with largest $\widehat\vs_i.$ The black line shows the corresponding $\widehat\vs_i$ after rescaling.}
	\label{fig:barplot citation}
\end{figure}

Furthermore, we give two examples of node comparison with respect to their in-status. Specifically, in each example, we choose a node $i$ and compare it with a randomly chosen subset $\SSSS = \{i_{k}:k=1,...,50\}$ with respect to in-status. 
Figure~\ref{fig:in-status comparison} shows the in-degree of nodes in $\{i\}\cup\SSSS$ after reordering, where each grey bar above the x-axis shows the positive in-degree and the red bar below the x-axis shows the negative in-degree. The solid line shows the estimated in-status for the 51 nodes after rescaling, where the blue triangle represents node $i,$ and the circles represent nodes in $\SSSS.$ The red and green points represent those nodes whose in-status are significantly different from that of node $i$ for the corresponding individual testing, whereas green points also indicate statistical significance for multiple testing. In the left panel, it is clear that the leftmost 14 nodes are significantly less attractive than node $i$ whereas the rightmost 16 nodes are significantly more attractive. It is further shown that the in-status of the rightmost 13 nodes are significantly different from that of node $i$ simultaneously. In the right panel, node $i$ is chosen to be Christopher Manning, and it is evident that his in-status is significantly higher than that of all nodes in $\SSSS$ simultaneously.

\begin{figure}[!htb]
	\centering
	\includegraphics[scale = 0.8]{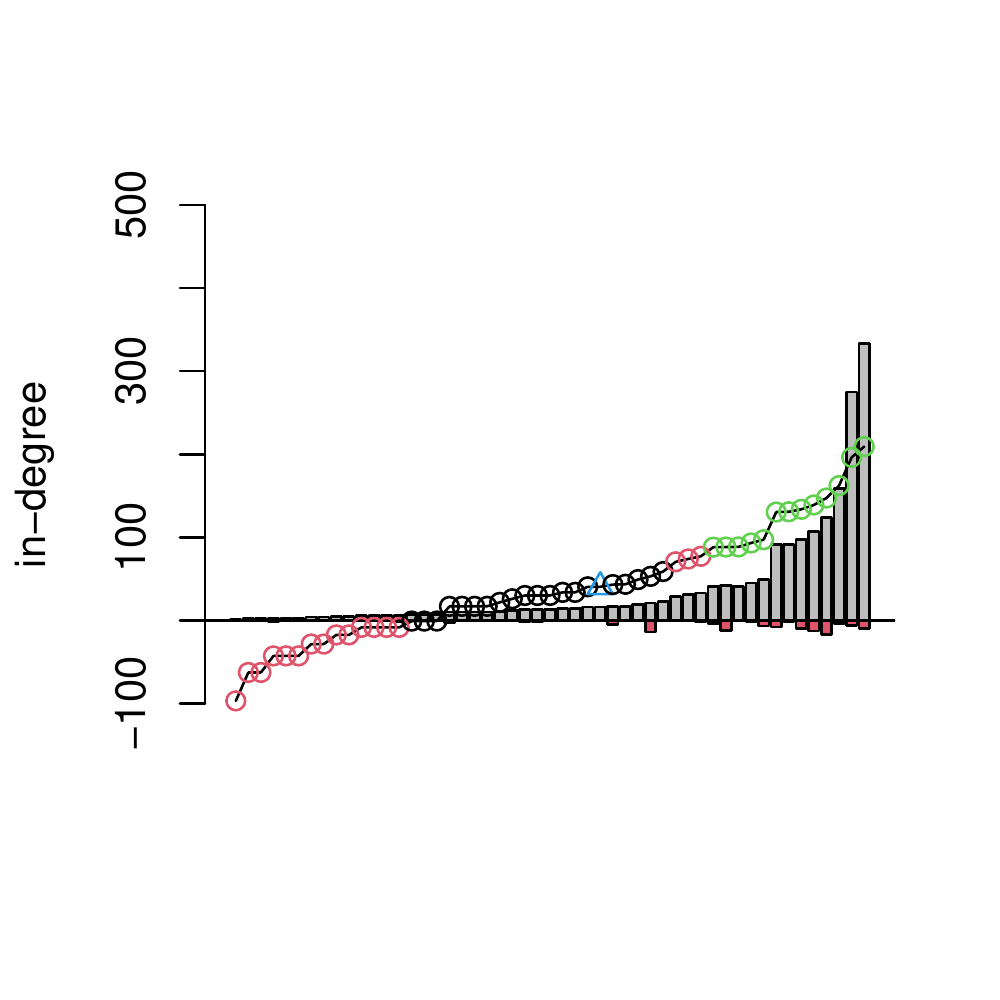}
	\includegraphics[scale = 0.8]{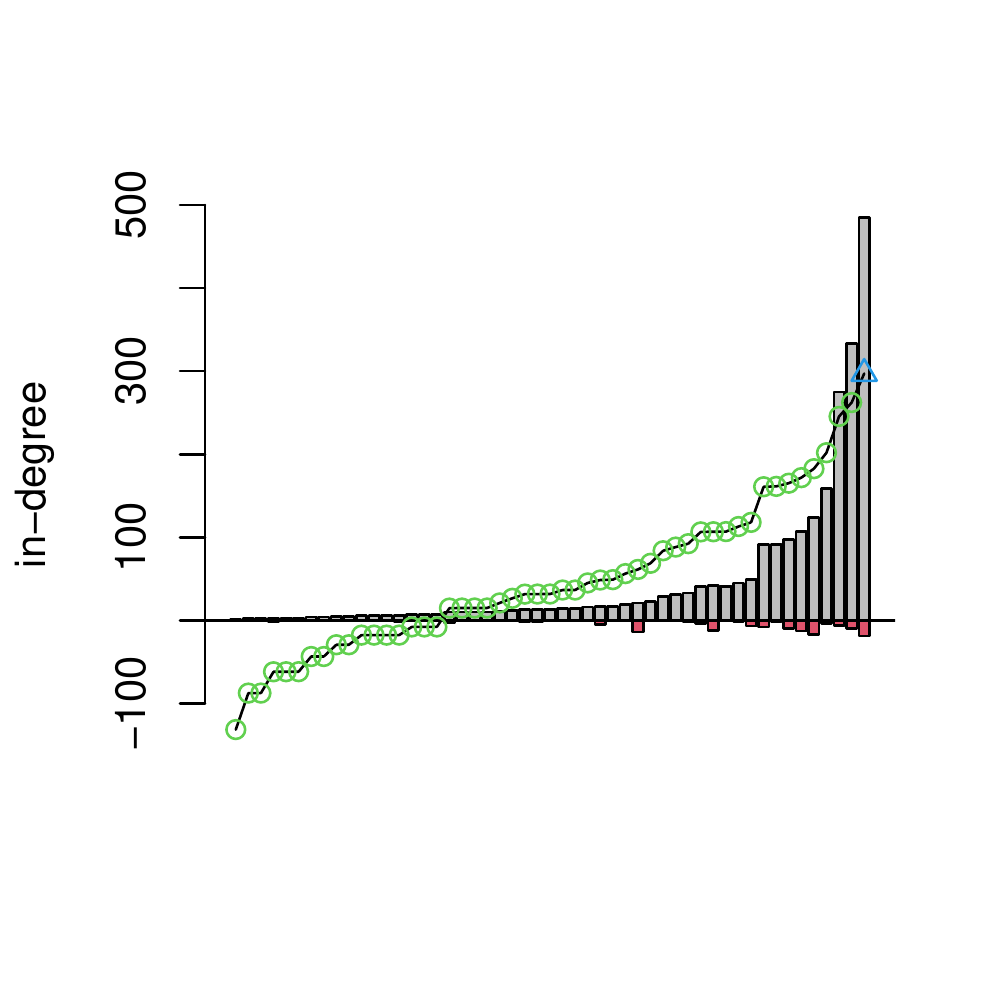}
	\vspace{-1.5cm}
	\caption{The estimated $\widehat\vc_j$ for $j\in\{i\}\cup\SSSS$ after rescaling in two examples, where the blue triangle represents node $i,$ and the circles represents nodes in $\SSSS.$  Both red and green points indicate statistical significance at level 95\% in comparison with node $i,$ whereas green points also indicate statistical significance for multiple testing. The bars further show the positive/negative in-degrees for these nodes.}
	\label{fig:in-status comparison}
\end{figure}

\section{Conclusion}

This paper proposes a general signed $\beta$-model for directed signed network, which embeds each node with two latent factors for in-status and out-status and includes the standard $\beta$-model and the directed $\beta$-model as special cases. We develop an efficient one-step estimation algorithm for parameter estimation and inferential procedure for node comparison, which leads to a computationally feasible estimation with asymptotic estimation efficiency. Note that the proposed model builds upon the mutual independence between the latent variables $z_{ij}^+$ and $z_{ij}^-$, it is thus of interest to relax this assumption and consider correlated latent variables. This will bring challenges in both computation and theoretical derivation, as the resultant likelihood function can be more complex than \eqref{eq:likelihood} and the current one-step estimation procedure is no longer applicable. Further, it is also interesting to extend the current procedure to analyze extremely sparse signed networks with exploitation of the regularization method \citep{zhang20212}. We leave it for future investigation.

\section*{Acknowledgment} 

This research is supported in part by HK RGC Grants GRF-11304520, GRF-11301521, GRF-11311022, and CUHK Startup Grant 4937091.

\appendix

\section*{Appendix A: notations and necessary lemmas}


For two negative sequences $a_n$ and $b_n$, let $a_n\lesssim_P b_n$ mean there exist positive constants $c_1,c_2$ such that $\Pr(a_n\geq c_1b_n) \leq c_2n^{-1}$ for $n$ large enough. With slight abuse of notations, we rewrite $\ttt = (\MS^\top,\MC^\top,\beta_n)^\top$ with $\theta_{2n} = \beta_n=0$, and define $\partial^2 l(\ttt;\sssp)/\partial\vs_i\partial\vc_n =  l''_{in}(\vs_i+\vc_n;\ssp_i)$ for $i \in [n-1]$ and $\partial^2 l(\ttt;\sssp)/\partial\vs_n\partial\vc_n = 0$. Further define $$
\frac{\partial l(\ttt;\sssp)}{\partial \vc_n} = \sum_{i=1}^{n-1} l'_{in}(\vs_i+\vc_n;\ssp_i)~~\text{and}~~\frac{\partial^2 l(\ttt;\sssp)}{\partial \vc_n^2} = \sum_{i=1}^{n-1} l''_{in}(\vs_i+\vc_n;\ssp_i).
$$ Similarly, $\check \ttt$, $\widehat \ttt$ and $\ttt^*$ are all augmented with an additional 0 entry, and the corresponding partial derivatives of $l(\check\ttt;\sssp)$, $l(\widehat\ttt;\sssp)$ and $l(\ttt^*;\sssp)$ are defined as above. Also, we denote
$$
\III_{ij} = \frac{\partial^2 l(\ttt^*;\sssp)}{\partial\vs_i\partial\vc_j}-\EEE\left[\frac{\partial^2 l(\ttt^*;\sssp)}{\partial\vs_i\partial\vc_j}\right],~\mbox{for}~i,j \in [n],
$$ 
and $\III_{ii} = 0$ for $i \in [n]$. 
Define  $d_i = \sum_{k=1,k\neq i}^n y_{ik}$ for any $i \in [n]$, $b_j = \sum_{k=1,k\neq j}^n y_{kj}$ for any $j \in [n]$, and $(g_i)_{i=1}^{2n} = (d_1,...,d_n,b_1,...,b_{n})^\top$. Further, for any $i,j\in[n],$ define 
$$
g_{i\backslash j} = g_i - y_{ij} = \sum_{k=1,k\neq i,j}^n y_{il},~\mbox{and}~g_{n+j\backslash i} = g_{n+j} - y_{ij} = \sum_{k=1,k\neq i,j}^n y_{kj}.
$$ 

Recall the definitions in \eqref{eq:u def}, \eqref{eq:uvw def}, \eqref{eq:u def2} and \eqref{eq:delta def}: $$
 \begin{aligned}
 	v_i =& \sum_{k=1,k\neq i}^n \frac{(1+\ssp_i)e^{\vs_i^*+\vc_k^*}}{(1+e^{\vs_i^*+\vc_k^*})^2},~&&w_i = \sum_{k=1,k\neq i}^n\var(y_{ik}),&&~\mbox{for}~i \in [n],\\
 	v_{n+j} =& \sum_{k=1,k\neq j}^n \frac{(1+\ssp_k)e^{\vs_k^*+\vc_j^*}}{(1+e^{\vs_k^*+\vc_j^*})^2},~&&w_{n+j}=\sum_{k=1,k\neq j}^n\var(y_{kj}),&&~\mbox{for}~j \in [n],
\end{aligned}
$$ $$
	 u_i = -\frac{\partial^2 l(\ttt^*;\sssp)}{\partial\theta_i^2},~~~~\check u_i = -\frac{\partial^2 l(\check\ttt;\sssp)}{\partial\theta_i^2},~~~~ \widehat u_{i} = -\frac{\partial^2 l(\widehat\ttt;\sssp)}{\partial\theta_i^2},~~\mbox{for}~i\in[2n].
$$ and $\widehat\delta_{ij}^{2} = \widehat u_i^{-1} + \widehat u_j^{-1},~(\delta_{ij}^*)^{2} = u_i^{-1} + (u_j^*)^{-1}$.
Below we list some necessary lemmas, and their proof are provided in Appendix C.

\begin{lemma}\label{lem:var}
There exists a positive constant $c$ such that $$
\begin{aligned}
c^{-1}ne^{-2\|\ttt^*\|_{\infty}}&\leq\min_{1\leq i\leq 2n} u_i\leq \max_{1\leq i\leq 2n} u_i\leq cn,\\
c^{-1}ne^{-2\|\ttt^*\|_{\infty}}&\leq\min_{1\leq i\leq 2n} v_i \leq \max_{1\leq i\leq 2n} v_i \leq cn,\\
c^{-1}ne^{-2\|\ttt^*\|_{\infty}}&\leq\min_{1\leq i\leq 2n} w_i \leq \max_{1\leq i\leq 2n} w_i \leq cn.
\end{aligned}
$$
Further, there exists a constant $\epsilon>0$ such that for any $\ttt,\sssp$ satisfying $\|\ttt-\ttt^*\|_{\infty}\leq \epsilon$ and $\|\sssp-\sssp^*\|_{\infty}\leq \epsilon$,
we have $$
n^{-1} \max_{1\leq i\leq 2n} \left| \frac{\partial^2 l(\ttt;\sssp)}{\partial \theta_i^2} - \frac{\partial^2 l(\ttt^*;\sssp)}{\partial \theta_i^2} \right| = O(\|\ttt-\ttt^*\|_{\infty}),
$$ $$
n^{-1} \max_{1\leq i\leq 2n} \left| \frac{\partial^2 l(\ttt;\sssp)}{\partial \theta_i^2} - \frac{\partial^2 l(\ttt;\sssp^*)}{\partial \theta_i^2} \right| = O(\|\sssp-\sssp^*\|_{\infty}).
$$
\end{lemma}

\begin{lemma}\label{lem:concentrate}
It holds true that 
\begin{equation}\label{eq:con1}
\Pr\left(\max_{1\leq i\leq 2n}|g_i - \EEE g_i| \leq \sqrt{4(n-1)\log (n-1)}\right) \geq 1-\frac{4n}{(n-1)^2}.
\end{equation}
Further, there exists a constant $c>0$ such that for $n$ large enough, we have 
\begin{equation}\label{eq:con2}
\Pr\left(\max_{1\leq i\leq 2n}\left|\frac{\partial l(\ttt^*;\sssp)}{\partial\theta_i}\right| \leq \sqrt{16(n-1)\log (n-1)} + cn\|\sssp-\sssp^*\|_{\infty}\right) \geq 1-\frac{4n}{(n-1)^2},
\end{equation} 
\begin{equation}\label{eq:con3}
\Pr\left(\max_{1\leq i\leq 2n} \left | \frac{\partial^2 l(\ttt^*;\sssp)}{\partial \theta_i^2} +u_i \right | \leq c\sqrt{(n-1)\log (n-1)}\right) \geq 1-\frac{4n}{(n-1)^2}.
\end{equation}
\end{lemma}

\begin{lemma}\label{lem:initial consis}
Under the conditions of Theorem~\ref{thm:initial}, with probability at least $1-4n/(n-1)^2$, it holds true that \eqref{eq:est1} has a unique solution $\check\ttt,$ and
$$
\|\check\ttt - \ttt^*\|_{\infty} \lesssim e^{6\|\ttt^*\|_{\infty}}\sqrt{\frac{\log n}{n}}.
$$ 
\end{lemma}

By Lemmas \ref{lem:var}-\ref{lem:initial consis}, 
we have $$
\max_{1\leq i\leq 2n} \left|\check u_i - u_i \right| \leq \max_{1\leq i\leq 2n} \left|\frac{\partial^2 l(\check\ttt;\sssp)}{\partial \theta_i^2} - \frac{\partial^2 l(\ttt^*;\sssp)}{\partial \theta_i^2} \right| + \max_{1\leq i\leq 2n} \left|\frac{\partial^2 l(\ttt^*;\sssp)}{\partial \theta_i^2} + u_i \right| \lesssim_P e^{6\|\ttt^*\|_{\infty}}\sqrt{n\log n},
$$
and \begin{equation}\label{eq:inverse}
\max_{1\leq i\leq 2n} \left|\check u_i^{-1} - (u_i\\
^*)^{-1} \right| = \max_{1\leq i\leq 2n}\frac{ \left|\check u_i - u_i \right|}{\left|u_i\check u_i\right|} \lesssim_P e^{10\|\ttt^*\|_{\infty}} \frac{\sqrt{\log n}}{n^{3/2}}.
\end{equation}

\begin{lemma}\label{lem:initial decomp}
Under the conditions of Theorem~\ref{thm:initial},
it holds true that 
$$
\begin{aligned}
\check \vs_i - \vs_i^* &= v_i^{-1}(g_i-\EEE g_i) + v_{2n}^{-1}(g_{2n} - \EEE g_{2n}) + \epsilon_{i},&&~~~~i=1,...,n,\\
\check \vc_j - \vc_j^* &= v_{n+j}^{-1}(g_{n+j}-\EEE g_{n+j}) - v_{2n}^{-1}(g_{2n} - \EEE g_{2n}) + \epsilon_{n+j},&&~~~~j=1,...,n-1,
\end{aligned}
$$ 
where $\epsilon_i$ satisfies that $$
P\Big(\max_{1\leq i\leq 2n-1}|\epsilon_i| \lesssim \frac{e^{18\|\ttt^*\|_{\infty}}\log n}{n} \Big) \ge 1-\frac{4n}{(n-1)^2}.$$
\end{lemma}

\begin{lemma}\label{lem:multi concentrate12}
	Under the conditions of Theorem~\ref{thm:initial},
it holds true that 
	\begin{equation}\label{eq:multi concentrate1}
	\begin{aligned}
	&\Pr\left(\max_{1\leq i\leq n}\left| \sum_{l=1}^{n-1}  u_i^{-1}\EEE\left[\frac{\partial^2l(\ttt^*;\sssp)}{\partial\vs_i\partial\vc_l}\right] \sqrt{n} v_{n+l}^{-1}(g_{n+l}-\EEE g_{n+l})  \right| \lesssim e^{4\|\ttt^*\|_{\infty}}
\sqrt{\frac{\log n}{n}}\right) \geq 1-\frac{2}{n},\\
	&\Pr\left(\max_{1\leq j\leq n}\left| \sum_{k=1}^n  u_{n+j}^{-1}\EEE\left[\frac{\partial^2 l(\ttt^*;\sssp)}{\partial\vs_k\partial\vc_j}\right] \sqrt{n} v_k^{-1}(g_k-\EEE g_k) \right| \lesssim e^{4\|\ttt^*\|_{\infty}}\sqrt{\frac{\log n}{n}}\right)\geq 1-\frac{2}{n},
	\end{aligned}
	\end{equation} and 
	\begin{equation}\label{eq:multi concentrate2}
	\begin{aligned}
	&\Pr\left(\max_{1\leq i\leq n}\left| \sum_{l=1}^{n-1}  u_i^{-1}\III_{il}\sqrt{n} v_{n+l}^{-1}(g_{n+l\backslash i}-\EEE g_{n+l\backslash i})  \right| \lesssim e^{4\|\ttt^*\|_{\infty}}\sqrt{\frac{\log n}{n}}\right)\geq 1-\frac{2}{n},\\
	&\Pr\left(\max_{1\leq j\leq n}\left| \sum_{k=1}^n  u_{n+j}^{-1}\III_{kj} \sqrt{n}v_k^{-1}(g_{k\backslash j}-\EEE g_{k\backslash j}) \right| \lesssim e^{4\|\ttt^*\|_{\infty}}\sqrt{\frac{\log n}{n}}\right)\geq 1-\frac{2}{n}.
	\end{aligned}
	\end{equation}
\end{lemma}

\begin{lemma}\label{lem:multi cent}
Under the conditions of Theorem~\ref{thm:initial},
it holds true that 
{\footnotesize $$
\max_{1\leq i\leq n}\left|\sum_{k=1}^n  u_{2n}^{-1}\EEE\left[\frac{\partial^2 l(\ttt^*;\sssp)}{\partial\vs_k\partial\vc_n}\right] \sqrt{n}(\check\vs_k - \vs_k^*) + \sum_{l=1}^{n-1}  u_i^{-1}\EEE\left[\frac{\partial^2l(\ttt^*;\sssp)}{\partial\vs_i\partial\vc_l}\right] \sqrt{n}(\check\vc_l - \vc_l^*)\right| \lesssim_P \frac{e^{20\|\ttt^*\|_{\infty}}\log n}{\sqrt{n}},
$$} and {\footnotesize $$
\max_{1\leq j\leq n-1}\left|\sum_{k=1}^n \sqrt{n} (\check\vs_k - \vs_k^*) \left\{ u_{n+j}^{-1} \EEE\left[\frac{\partial^2 l(\ttt^*;\sssp)}{\partial\vs_k\partial\vc_j}\right] - u_{2n}^{-1} \EEE\left[\frac{\partial^2 l(\ttt^*;\sssp)}{\partial\vs_k\partial\vc_n}\right] \right\} \right| \lesssim_P \frac{e^{20\|\ttt^*\|_{\infty}}\log n}{\sqrt{n}}.
$$} 
\end{lemma}

\begin{lemma}\label{lem:multi cent2}
	Under the conditions of Theorem~\ref{thm:initial},
it holds true that $$
	\max_{1\leq i\leq n}\left|\sum_{k=1}^n  u_{2n}^{-1}\III_{kn} \sqrt{n}(\check\vs_k - \vs_k^*) + \sum_{l=1}^{n-1}  u_i^{-1}\III_{il} \sqrt{n}(\check\vc_l - \vc_l^*)\right| \lesssim_P \frac{e^{20\|\ttt^*\|_{\infty}}\log n}{\sqrt{n}},
	$$ and $$
	\max_{1\leq j\leq n-1}\left|\sum_{k=1}^n \sqrt{n} (\check\vs_k - \vs_k^*) \left\{ u_{n+j}^{-1} \III_{kj} - u_{2n}^{-1} \III_{kn}\right\} \right| \lesssim_P \frac{e^{20\|\ttt^*\|_{\infty}}\log n}{\sqrt{n}}.
	$$ 
\end{lemma}

\section*{Appendix B:  main proofs}

\begin{proof}[\bf Proof of Theorem \ref{thm:initial}]
	The asymptotic bound for $\|\check\ttt-\ttt^*\|_{\infty}$ in \eqref{eq:consis0} follows from Lemma \ref{lem:initial consis} immediately. To obtain asymptotic normality, it follows from Lemma \ref{lem:initial decomp} that
	\begin{equation}\label{eqn:alpha}
	\begin{aligned}
	&\max_{1\leq i\leq n} \left| (\check \vs_i - \vs_i^*) - \left\{v_i^{-1}(g_i-\EEE g_i) + v_{2n}^{-1}(g_{2n} - \EEE g_{2n}) \right\} \right| \lesssim_P \frac{e^{18\|\ttt^*\|_{\infty}}\log n}{n}, \\
	&\max_{1\leq j\leq n-1} \left| (\check \vc_j - \vc_j^*) - \left\{v_{n+j}^{-1}(g_{n+j}-\EEE g_{n+j}) - v_{2n}^{-1}(g_{2n} - \EEE g_{2n}) \right\} \right| \lesssim_P \frac{e^{18\|\ttt^*\|_{\infty}}\log n}{n}.
	\end{aligned}
	\end{equation} 
	Define $\qq\in\R^{2n-1}$ with $q_i = v_i^{-1}(g_i-\EEE g_i) + v_{2n}^{-1}(g_{2n} - \EEE g_{2n})$ for $i=1,...,n$, and $q_{n+j} = v_{n+j}^{-1}(g_{n+j}-\EEE g_{n+j}) - v_{2n}^{-1}(g_{2n} - \EEE g_{2n})$ for $j=1,...,n-1$.	
	It is not difficult to verify that for any fixed $d$, $\qq_{[1:d]}$ is asymptotically multivariate normal with mean zero and covariance matrix $\SG_{[1:d,1:d]}$. By \eqref{eqn:alpha} and Lemma~\ref{lem:var}, we have $$
	\|(\SG_{[1:d,1:d]})^{-\frac{1}{2}}(\check\ttt-\ttt^* - \qq)_{[1:d]}\|_{\infty} \lesssim_P \sqrt{n}e^{\|\ttt^*\|_{\infty}}\frac{e^{18\|\ttt^*\|_{\infty}}\log n}{n}= \frac{e^{19\|\ttt^*\|_{\infty}}\log n}{\sqrt{n}}  = o(1),
	$$ which completes the proof of Theorem \ref{thm:initial}.
\end{proof}

\begin{proof}[\bf Proof of Theorem~\ref{thm:normality}]
By Taylor's expansion, there exists a $\widetilde\ttt$ such that $$
\frac{\partial l(\check\ttt;\sssp)}{\partial\ttt} = \frac{\partial l(\ttt^*;\sssp)}{\partial\ttt} + \frac{\partial^2 l(\widetilde\ttt;\sssp)}{\partial\ttt^2}(\check\ttt-\ttt^*),
$$ where $\widetilde \theta_i$ lies between $\theta_i^*$ and $\check\theta_i$ for $i \in [2n-1]$ and $\widetilde \theta_{2n}=\theta_{2n}^*=\check\theta_{2n}=0$. It is not difficult to verify that $\partial l(\ttt;\sssp)/\partial\vc_n = \sum_{k=1}^{n}\partial l(\ttt;\sssp)/\partial\vs_k - \sum_{l=1}^{n-1}\partial l(\ttt;\sssp)/\partial\vc_l$. 
Then, we have $$
\begin{aligned}
	\widehat\vs_i - \vs_i^* =&(\check\vs_i - \vs_i^*) + \check u_i^{-1} \frac{\partial l(\check\ttt;\sssp)}{\partial \vs_i} + \check u_{2n}^{-1}\frac{\partial l(\check\ttt;\sssp)}{\partial\vc_n} \\
	=& (\check\vs_i - \vs_i^*) + \check u_i^{-1}\left\{ \frac{\partial l(\ttt^*;\sssp)}{\partial\vs_i} + \frac{\partial^2 l(\widetilde\ttt;\sssp)}{\partial\vs_i^2} (\check\vs_i-\vs_i^*) + \sum_{j=1}^{n-1}  \frac{\partial^2 l(\widetilde\ttt;\sssp)}{\partial\vs_i\partial\vc_j} (\check\vc_j-\vc_j^*)  \right\} \\
	&+\check u_{2n}^{-1} \left\{ \frac{\partial l(\ttt^*;\sssp)}{\partial\vc_n} + \sum_{k=1}^{n}  \frac{\partial^2 l(\widetilde\ttt;\sssp)}{\partial\vs_k\partial\vc_n} (\check\vs_k-\vs_k^*) \right\}\\
	=& \left\{ \check u_i^{-1} \frac{\partial l(\ttt^*;\sssp)}{\partial \vs_i} + \check u_{2n}^{-1}\frac{\partial l(\ttt^*;\sssp)}{\partial \vc_n} \right\} + (\check\vs_i-\vs_i^*)\left\{ 1 + \check u_i^{-1}\frac{\partial^2 l(\widetilde\ttt;\sssp)}{\partial\vs_i^2} \right\} \\
	&+ \left\{\check u_{2n}^{-1} \sum_{k=1}^n \frac{\partial^2 l(\widetilde\ttt;\sssp)}{\partial\vs_k\partial\vc_n} (\check\vs_k - \vs_k^*) + \check u_i^{-1}\sum_{j=1}^{n-1}\frac{\partial^2l(\widetilde\ttt;\sssp)}{\partial\vs_i\partial\vc_j} (\check\vc_j - \vc_j^*)\right\}\\
	=:& I_1^{(i)}+I_2^{(i)}+I_3^{(i)}.
\end{aligned}
$$

For $I_1^{(i)}$, we have
$$
\begin{aligned}
	I_1^{(i)} =& \left\{ u_i^{-1} \frac{\partial l(\ttt^*;\sssp)}{\partial \vs_i} + u_{2n}^{-1}\frac{\partial l(\ttt^*;\sssp)}{\partial \vc_n} \right\} + \left[ -u_i^{-1} + \check u_i^{-1} \right]\frac{\partial l(\ttt^*;\sssp)}{\partial \vs_i} + \left[ \check u_{2n}^{-1} - u_{2n}^{-1} \right]\frac{\partial l(\ttt^*;\sssp)}{\partial \vc_n}.
\end{aligned}
$$ 
The first term can be rewritten as
$$
\begin{aligned}
& u_i^{-1} \frac{\partial l(\ttt^*;\sssp)}{\partial \vs_i} + u_{2n}^{-1}\frac{\partial l(\ttt^*;\sssp)}{\partial \vc_n}\\
=& \left\{u_i^{-1} \frac{\partial l(\ttt^*;\sssp)}{\partial \vs_i} + u_{2n}^{-1}\sum_{k=1,k\neq i}^{n-1}l'_{kn}(\vs_k^*+\vc_n^*;\ssp_k)\right\} + u_{2n}^{-1}l'_{in}(\vs_i^*+\vc_n^*;\ssp_i)\\
=:&  I_{11}^{(i)} + I_{12}^{(i)}
\end{aligned}
$$
By central limit theorem,  $I_{11}^{(i)}$ is the sum of two independent variables which are asymptotic normal with variances $u_i^{-1}$ and $u_{2n}^{-1},$ respectively. Therefore, $I_{11}^{(i)}$ is also asymptotic normal with variance $u_i^{-1} + u_{2n}^{-1}.$ Also, Lemma~\ref{lem:var} implies that $\max_{1\leq i\leq n}|I_{12}^{(i)}| = O(e^{2\|\ttt^*\|_{\infty}}n^{-1}) = o(n^{-1/2}).$ Furthermore, it follows from Lemma \ref{lem:concentrate} and \eqref{eq:inverse} that 
$$
\max_{1\leq i\leq n} \left |  \left[ u_i^{-1} - \check u_i^{-1} \right]\frac{\partial l(\ttt^*;\sssp)}{\partial \vs_i} + \left[ \check u_{2n}^{-1} - u_{2n}^{-1} \right]\frac{\partial l(\ttt^*;\sssp)}{\partial \vc_n} \right | \lesssim_P \frac{e^{10\|\ttt^*\|_{\infty}}\log n}{n},
$$ which is of order $o_p(n^{-1/2})$.
Therefore, we have $\sqrt{u_i^{-1} + u_{2n}^{-1}} I_1^{(i)} \to N(0,1)$ as $n$ diverges.

For $I_2^{(i)},$ it follows from Lemmas~\ref{lem:var} and \ref{lem:initial consis} that
$$
	\max_{1\leq i\leq n}|I_2^{(i)}| \leq \max_{1\leq i\leq n}|\check\vs_i - \vs_i^*| \left| \frac{\check u_i + \frac{\partial^2 l(\widetilde \ttt)}{\partial \vs_i^2}}{\check u_i} \right| \lesssim_P e^{6\|\ttt^*\|_{\infty}}\sqrt{\frac{\log n}{n}}\frac{e^{6\|\ttt^*\|_{\infty}}\sqrt{n\log n}}{e^{-2\|\ttt^*\|_{\infty}}n} = \frac{e^{14\|\ttt^*\|_{\infty}}\log n}{n},
$$ which is of order $o_p(n^{-1/2})$.

For $I_3^{(i)},$ we have $$
\begin{aligned}
	I_3^{(i)} =& \left\{\sum_{k=1}^n \check u_{2n}^{-1}\EEE\left[\frac{\partial^2 l(\ttt^*;\sssp)}{\partial\vs_k\partial\vc_n}\right] (\check\vs_k - \vs_k^*) + \sum_{j=1}^{n-1}  \check u_{i}^{-1}\EEE\left[\frac{\partial^2l(\ttt^*;\sssp)}{\partial\vs_i\partial\vc_j}\right] (\check\vc_j - \vc_j^*)\right\}\\
	&+ \left\{ \sum_{k=1}^n \check u_{2n}^{-1}\III_{kn} (\check\vs_k - \vs_k^*) + \sum_{j=1}^{n-1}  \check u_{i}^{-1}\III_{ij}(\check\vc_j - \vc_j^*) \right\} \\
	&+ \left\{ \sum_{k=1}^n \check u_{2n}^{-1}\left(\frac{\partial^2 l(\widetilde\ttt;\sssp)}{\partial\vs_k\partial\vc_n}-\frac{\partial^2 l(\ttt^*;\sssp)}{\partial\vs_k\partial\vc_n}\right) (\check\vs_k - \vs_k^*) + \sum_{j=1}^{n-1}  \check u_{i}^{-1}\left(\frac{\partial^2l(\widetilde\ttt;\sssp)}{\partial\vs_i\partial\vc_j}-\frac{\partial^2l(\ttt^*;\sssp)}{\partial\vs_i\partial\vc_j}\right)(\check\vc_j - \vc_j^*) \right\} \\
	=:& I_{31}^{(i)} + I_{32}^{(i)} + I_{33}^{(i)}. 
\end{aligned}
$$ 

Note that $|I_{31}^{(i)}|$ can be further bounded as 
$$
\begin{aligned}
	|I_{31}^{(i)}| &\le \left | \sum_{k=1}^n  u_{2n}^{-1}\EEE\left[\frac{\partial^2 l(\ttt^*;\sssp)}{\partial\vs_k\partial\vc_n}\right] (\check\vs_k - \vs_k^*) + \sum_{j=1}^{n-1}  u_i^{-1}\EEE\left[\frac{\partial^2l(\ttt^*;\sssp)}{\partial\vs_i\partial\vc_j}\right](\check\vc_j - \vc_j^*)\right | \\
	&+ \left|\sum_{k=1}^n  \left[\check u_{2n}^{-1}-u_{2n}^{-1}\right]\EEE\left[\frac{\partial^2 l(\ttt^*;\sssp)}{\partial\vs_k\partial\vc_n}\right](\check\vs_k - \vs_k^*) + \sum_{j=1}^{n-1}  \left[\check u_{i}^{-1}-u_i^{-1}\right]\EEE\left[\frac{\partial^2l(\ttt^*;\sssp)}{\partial\vs_i\partial\vc_j}\right](\check\vc_j - \vc_j^*)\right|.
\end{aligned}
$$ 

By Lemma \ref{lem:multi cent}, the first term is upper bounded by $e^{20\|\ttt^*\|_{\infty}}\log n/n$ in probability, and by Lemma~\ref{lem:initial consis} and \eqref{eq:inverse}, the second term is upper bounded by $e^{16\|\ttt^*\|_{\infty}}\log n/n$ in probability. Therefore, we have $\max_{1\leq i\leq n}|I_{31}^{(i)}| \lesssim_P e^{20\|\ttt^*\|_{\infty}}\log n / n$.

For $I_{32}^{(i)},$ we have $$
\begin{aligned}
	|I_{32}^{(i)}| \le & \left | \sum_{k=1}^n u_{2n}^{-1}\III_{kn} (\check\vs_k - \vs_k^*) + \sum_{j=1}^{n-1}  u_i^{-1}\III_{ij}(\check\vc_j - \vc_j^*)\right |\\
	& + \left | \sum_{k=1}^n \left[\check u_{2n}^{-1} - u_{2n}^{-1} \right]\III_{kn} (\check\vs_k - \vs_k^*) + \sum_{j=1}^{n-1}  \left[\check u_i^{-1} - u_i^{-1}\right]\III_{ij}(\check\vc_j - \vc_j^*)\right |.
\end{aligned}
$$ 
By Lemma \ref{lem:multi cent2}, the first term is upper bounded by $e^{20\|\ttt^*\|_{\infty}}\log n/n$ in probability, and by Lemma \ref{lem:initial consis} and \eqref{eq:inverse},the second term is upper bounded by $e^{16\|\ttt^*\|_{\infty}}\log n/n$ in probability. Therefore, we have $\max_{1\leq i\leq n}|I_{32}^{(i)}| \lesssim_P e^{20\|\ttt^*\|_{\infty}}\log n / n$.

For $I_{33}^{(i)},$ it follows from Lemma \ref{lem:var}, Lemma \ref{lem:initial consis} and \eqref{eq:inverse} that $\max_{1\leq i\leq n}|I_{33}^{(i)}| \lesssim_{P} e^{14\|\ttt^*\|_{\infty}}\log n / n$. Then, combining the upper bounds for each term and Lemma~\ref{lem:concentrate} yields that
\begin{equation}\label{eq:vs residual}
\max_{1\leq i\leq n}\left|(\widehat \vs_i - \vs_i^*) -\left\{ u_i^{-1}\frac{\partial l(\ttt^*;\sssp)}{\partial \vs_i} + u_{2n}^{-1}\frac{\partial l(\ttt^*;\sssp)}{\partial\vc_n} \right\}  \right| \lesssim_P \frac{e^{20\|\ttt^*\|_{\infty}}\log n}{n},
\end{equation}
where $u_i^{-1}\frac{\partial l(\ttt^*;\sssp)}{\partial \vs_i} + u_{2n}^{-1}\frac{\partial l(\ttt^*;\sssp)}{\partial\vc_n}$ is asymptotic normal with variance $u_i^{-1} + u_{2n}^{-1}.$

Next, we turn to bound $\widehat\vc_j - \vc_j^*$, where
$$
\begin{aligned}
	\widehat\vc_j - \vc_j^*=& (\check\vc_j - \vc_j^*) +\check u_{n+j}^{-1} \frac{\partial l(\check\ttt;\sssp)}{\partial \vc_j} - \check u_{2n}^{-1}\frac{\partial l(\check\ttt;\sssp)}{\partial\vc_n}  \\
	=& (\check\vc_j - \vc_j^*) + \check u_{n+j}^{-1}\left\{ \frac{\partial l(\ttt^*;\sssp)}{\partial\vc_j} + \frac{\partial^2 l(\widetilde\ttt;\sssp)}{\partial\vc_j^2} (\check\vc_j-\vc_j^*) + \sum_{k=1}^{n}  \frac{\partial^2 l(\widetilde\ttt;\sssp)}{\partial\vs_k\partial\vc_j} (\check\vs_k-\vs_k^*)  \right\} \\
	&-\check u_{2n}^{-1}\left\{ \frac{\partial l(\ttt^*;\sssp)}{\partial\vc_n} + \sum_{k=1}^{n}  \frac{\partial^2 l(\widetilde\ttt;\sssp)}{\partial\vs_k\partial\vc_n} (\check\vs_k-\vs_k^*) \right\}\\
	=& \left\{ \check u_{n+j}^{-1} \frac{\partial l(\ttt^*;\sssp)}{\partial \vc_j} - \check u_{2n}^{-1}\frac{\partial l(\ttt^*;\sssp)}{\partial \vc_n} \right\} + (\check\vc_j-\vc_j^*)\left\{ 1+\check u_{n+j}^{-1}\frac{\partial^2 l(\widetilde\ttt;\sssp)}{\partial\vc_j^2} \right\} \\
	&+ \left\{\sum_{k=1}^n (\check\vs_k - \vs_k^*)\left[ -\check u_{2n}^{-1} \frac{\partial^2 l(\widetilde\ttt;\sssp)}{\partial\vs_k\partial\vc_n} + \check u_{n+j}^{-1} \frac{\partial^2 l(\widetilde\ttt;\sssp)}{\partial\vs_k\partial\vc_j} \right]\right\}\\
	=:& J_1^{(j)}+J_2^{(j)}+J_3^{(j)}.
\end{aligned}
$$ 
Similarly as the case of $\widehat\vs_i - \vs_i^*,$ we have 
$$
J_1^{(j)} = \left\{ u_{n+j}^{-1}\frac{\partial l(\ttt^*;\sssp)}{\partial \vc_j} - u_{2n}^{-1}\frac{\partial l(\ttt^*;\sssp)}{\partial\vc_n} \right\} + \left[ -u_{n+j}^{-1} + \check u_{n+j}^{-1} \right]\frac{\partial l(\ttt^*;\sssp)}{\partial \vc_j} + \left[ -\check u_{2n}^{-1} + u_{2n}^{-1} \right]\frac{\partial l(\ttt^*;\sssp)}{\partial \vc_n}, 
$$
where it follows from Lemma \ref{lem:concentrate} and \eqref{eq:inverse} that 
$$
\max_{1\leq j \leq n-1} \left | \left[ -u_{n+j}^{-1} + \check u_{n+j}^{-1} \right]\frac{\partial l(\ttt^*;\sssp)}{\partial \vc_j} + \left[ - \check u_{2n}^{-1} + u_{2n}^{-1} \right]\frac{\partial l(\ttt^*;\sssp)}{\partial \vc_n} \right | \lesssim_P \frac{e^{10\|\ttt^*\|_{\infty}}\log n}{n}.
$$ 
Further, we also have
$$
	\max_{1\leq j\leq n-1}|J_2^{(j)}| \lesssim_P  \frac{e^{14\|\ttt^*\|_{\infty}}\log n}{n}, \ \mbox{and} \ \max_{1\leq j\leq n-1}|J_3^{(j)}| \lesssim_P  \frac{e^{20\|\ttt^*\|_{\infty}}\log n}{n}.
$$

Combing all the results together and Lemma~\ref{lem:concentrate}, 
we have \begin{equation}\label{eq:vc residual}
\left|(\widehat \vc_j - \vc_j^*) - \left\{ u_{n+j}^{-1}\frac{\partial l(\ttt^*;\sssp)}{\partial \vc_j} - u_{2n}^{-1}\frac{\partial l(\ttt^*;\sssp)}{\partial\vc_n} \right\} \right| \lesssim_P  \frac{e^{20\|\ttt^*\|_{\infty}}\log n}{n}
\end{equation}
where central limit theorem implies that $u_{n+j}^{-1}\frac{\partial l(\ttt^*;\sssp)}{\partial \vc_j} - u_{2n}^{-1}\frac{\partial l(\ttt^*;\sssp)}{\partial\vc_n}$ is asymptotic normal with variance $u_{n+j}^{-1} + u_{2n}^{-1}.$ 

Define $\rr\in\R^{2n-1}$, where $r_i = u_i^{-1}\frac{\partial l(\ttt^*;\sssp)}{\partial \vs_i} + u_{2n}^{-1}\frac{\partial l(\ttt^*;\sssp)}{\partial\vc_n}$ for $i\in[n]$, and $r_{n+j} = u_{n+j}^{-1}\frac{\partial l(\ttt^*;\sssp)}{\partial \vc_j} - u_{2n}^{-1}\frac{\partial l(\ttt^*;\sssp)}{\partial\vc_n}$ for $j\in[n-1]$. It is not difficult to verify that for any fixed $d$, $\rr_{[1:d]}$ is asymptotically multivariate normal with mean zero and covariance matrix $\HH_{[1:d,1:d]}$.
By \eqref{eq:vs residual} and \eqref{eq:vc residual}, we obtain $\|\widehat\ttt - \ttt^* - \rr\|_{\infty} = o_p(n^{-1/2})$, and thus complete the proof of asymptotical normality.

	By Lemma~\ref{lem:var} and \eqref{eq:con2} in Lemma~\ref{lem:concentrate}, we have $$
\max_{1\leq i\leq 2n} \left| u_i^{-1}\frac{\partial l(\ttt^*;\sssp)}{\partial \theta_i} \right| \lesssim_P e^{2\|\ttt^*\|_{\infty}}\sqrt{\frac{\log n}{n}}.
$$ Then, \eqref{eq:consis1} is obtained from \eqref{eq:vs residual} and \eqref{eq:vc residual}.
\end{proof}

\begin{proof}[\bf Proof of Proposition~\ref{prop:var compare}] Let $\bar\vs_i$ be the solution of the estimation equation 
	$$
	d_i - \sum_{k=1,k\neq i}^n g(\vc_j^*-\vs_i) = 0.
	$$ 
	It follows from Theorem~\ref{thm:initial} that the asymptotic variance of $\bar\vs_i$ is $w_i (v_i^*)^{-2}$. But by the Cramer-Rao bound, we also have 
	$$
	\var(\bar\vs_i) \geq \left ( -\EEE\left[\frac{\partial^2 l(\ttt^*)}{\partial\vs_i^2}\right] \right)^{-1}= -(u_i^*)^{-1}.
	$$ 
	By letting $n\to\infty,$ we get $w_i (v_i^*)^{-2} \geq |u_i^*|^{-1}$ for any $i \in [n]$. Similarly, we can also get $w_i (v_i^*)^{-2} \geq |u_i^*|^{-1}$ for any $i=n+1,...,2n.$ Therefore, for any $i \in [2n]$, we have
	$$
	\frac{\avar(\check\theta_i)}{\avar(\widehat\theta_i)} = \frac{w_i (v_i^*)^{-2}+w_{2n} (v_{2n}^*)^{-2}}{|u_i^*|^{-1}+|u_{2n}^*|^{-1}} \geq \min\left\{ \frac{w_i (v_i^*)^{-2}}{|u_i^*|^{-1}}, \frac{w_{2n} (v_{2n}^*)^{-2}}{|u_{2n}^*|^{-1}} \right\} \geq 1.
	$$ 
	This completes the proof of Proposition~\ref{prop:var compare}.
\end{proof}

\begin{proof}[\bf Proof of Theorem~\ref{thm:ranking}]

	By \eqref{eq:vs residual} and \eqref{eq:vc residual}, we have $$
	\begin{aligned}
		(\widehat\vs_i - \widehat\vs_j) - (\vs_i^*- \vs_j^*) &= \left\{ u_i^{-1}\frac{\partial l(\ttt^*;\sssp)}{\partial \vs_i} - (u_j^*)^{-1}\frac{\partial l(\ttt^*;\sssp)}{\partial \vs_j} \right\} + o_p(n^{-\frac{1}{2}}),\\
		(\widehat\vc_i - \widehat\vc_j) - (\vc_i^*- \vc_j^*) &= \left\{ (u_{n+i}^*)^{-1}\frac{\partial l(\ttt^*;\sssp)}{\partial \vc_i} - u_{n+j}^{-1}\frac{\partial l(\ttt^*;\sssp)}{\partial \vc_j} \right\} + o_p(n^{-\frac{1}{2}}),
	\end{aligned}
	$$ 
	which immediately imply \eqref{eq:rank normality}.

 By \eqref{eq:consis1}, similar to the derivation of \eqref{eq:inverse}, 
	we have 
	$$
	\max_{1\leq i\leq 2n} \left|\widehat u_i^{-1} - u_i^{-1} \right| \lesssim_P e^{6\|\ttt^*\|_{\infty}}\frac{\sqrt{\log n}}{n^{3/2}},
	$$ 
	which, together with Lemma~\ref{lem:concentrate}, implies that 
	\begin{equation}\label{eq:inverse2}
	\begin{aligned}
	& \max_{i\neq j}|(\delta_{ij}^*)^{-1} - \widehat\delta_{ij}^{-1}| =\max_{i\neq j}\left| \frac{\widehat\delta_{ij}^2 - (\delta_{ij}^*)^2}{\delta_{ij}^*\widehat\delta_{ij}(\widehat\delta_{ij}+\delta^*_{ij})} \right| \\ = & \max_{i\neq j}\left| \frac{[u_i^{-1}-\widehat u_i^{-1}] + [(u_{j}^*)^{-1} - \widehat u_j^{-1}]}{\delta_{ij}^*\widehat\delta_{ij}(\widehat\delta_{ij}+\delta^*_{ij})} \right| \lesssim_P \frac{e^{6\|\ttt^*\|_{\infty}}n^{-3/2}\sqrt{\log n}}{n^{-3/2}} = e^{6\|\ttt^*\|_{\infty}}\sqrt{\log n}.
	\end{aligned}
	\end{equation} 
	Simple algebra yields that
	\begin{align*}
		&\widehat\delta_{ij}^{-1} \big[(\widehat\vs_i - \widehat\vs_j) - (\vs_i^*- \vs_j^*)\big]\\
		= & (\delta_{ij}^*)^{-1} \left[(\widehat\vs_i - \widehat\vs_j) - (\vs_i^*- \vs_j^*)\right] + ( \widehat\delta_{ij}^{-1} - (\delta_{ij}^*)^{-1}) \left[(\widehat\vs_i - \widehat\vs_j) - (\vs_i^*- \vs_j^*)\right],
	\end{align*}
	where the first term converges to $N(0,1)$ in distribution following \eqref{eq:rank normality}, and the second one is bounded by $O_p(e^{8\|\ttt^*\|_{\infty}}\log n / \sqrt{n}) = o_p(1)$ following \eqref{eq:consis1} and \eqref{eq:inverse2}. The case for $\widehat\delta_{n+i,n+j}^{-1}\big[(\widehat\vc_i - \widehat\vc_j) - (\vc_i^*- \vc_j^*)\big]$ is similar. This completes the proof of Theorem \ref{thm:ranking}. \end{proof}


\begin{proof}[\bf Proof of Theorem~\ref{thm:fdr}]
	First, by the definition of FDR, we have
	\begin{equation}\label{eq:fdr decomp}
		\begin{aligned}
			& \FDR = \EEE\left[ \frac{1}{r}\sum_{k\in\SSSS_0}1_{\{p_k\leq\frac{\alpha r}{KL}\}}1_{\{r>0\}} \right] = \sum_{k\in\SSSS_0}\EEE\left[\frac{1}{r}1_{\{p_k\leq\frac{\alpha r}{KL}\}}1_{\{r>0\}}\right]\\
			=&\sum_{l=1}^{\infty}\frac{1}{l(l+1)}\sum_{k\in\SSSS_0}\EEE\left[1_{\{p_k\leq\frac{\alpha r}{KL}\}}1_{\{0<r\leq l\}}\right] \leq \sum_{l=1}^{\infty}\frac{1}{l(l+1)}\sum_{k\in\SSSS_0}\Pr\left(p_k\leq\frac{\alpha \min(l,K)}{KL}\right),
		\end{aligned}
	\end{equation}
	where the third equality is due to that for each $r_0>0$, $$
	\begin{aligned}
	\frac{1}{r_0}1_{\{p_k \leq\frac{\alpha r_0}{KL}\}}1_{\{r_0>0\}} =& 1_{\{p_k \leq\frac{\alpha r_0}{KL}\}}1_{\{r_0>0\}}  \sum_{l\geq r_0}\left(\frac{1}{l} - \frac{1}{l+1}\right) = 1_{\{p_k \leq\frac{\alpha r_0}{KL}\}}1_{\{r_0>0\}} \sum_{l= r_0}^{\infty}\frac{1}{l(l+1)}\\ =& 1_{\{p_k \leq\frac{\alpha r_0}{KL}\}} \sum_{l=1}^{\infty}\frac{1}{l(l+1)} 1_{\{0<r_0\leq l\}}.
	\end{aligned}
	$$
	It thus suffices to establish an upper bound for $\max_{k \in \SSSS_0} \Pr\left(p_k\leq\frac{\alpha \min(l,K)}{KL}\right).$
	
	Note that for any $k\in\SSSS_0,$ it follows from \eqref{eq:vs residual} that
	$$
	\widehat\vs_i-\widehat\vs_k = V_{ik} + (\omega_i-\omega_k)~~\text{with}~~ \max_{k\in\SSSS_0} |\omega_i - \omega_k| \lesssim_P e^{20\|\ttt^*\|_{\infty}}\frac{\log n}{n},
	$$ 
	where $V_{ik} =  u_i^{-1}\frac{\partial l(\ttt^*;\sssp)}{\partial \vs_i} - u_k^{-1}\frac{\partial l(\ttt^*;\sssp)}{\partial \vs_k}$ and $
	\omega_i = (\widehat \vs_i - \vs_i^*) -\left\{ u_i^{-1}\frac{\partial l(\ttt^*;\sssp)}{\partial \vs_i} + u_{2n}^{-1}\frac{\partial l(\ttt^*;\sssp)}{\partial\vc_n} \right\}.$ Then, we have			
	$$
		\begin{aligned}
			&\max_{k\in\SSSS_0}\Pr\left(p_k\leq\frac{\alpha \min(l,K)}{KL}\right)\\
			=& \max_{k\in\SSSS_0}\Pr\left(2\left[1 - \Phi\left(\widehat\delta_{ik}^{-1}|\widehat\vs_i-\widehat\vs_k|\right)\right] \leq \frac{\alpha \min(l,K)}{KL}\right)\\
			=&\max_{k\in\SSSS_0} \Pr\left(\widehat\delta_{ik}^{-1}\big |V_{ik} + (\omega_i-\omega_k) \big |
			\geq \Phi^{-1}\left(1-\frac{\alpha \min(l,K)}{2KL}\right)\right)\\
			\leq& \max_{k\in\SSSS_0}\Pr\left((\delta_{ik}^*)^{-1}| V_{ik} |
			\geq \Phi^{-1}\left(1-\frac{\alpha \min(l,K)}{2KL}\right) - \mu_{ik} \right) \\
			\leq&\max_{k\in\SSSS_0}\Pr\left((\delta_{ik}^*)^{-1}| V_{ik} |
			\geq \Phi^{-1}\left(1-\frac{\alpha \min(l,K)}{2KL}\right) - \frac{e^{20\|\ttt^*\|_{\infty}}(\log n)^2}{\sqrt{n}}\right) + \max_{k\in\SSSS_0}\Pr \Big (\mu_{ik} \geq \frac{e^{20\|\ttt^*\|_{\infty}}(\log n)^2}{\sqrt{n}} \Big ),
		\end{aligned}
	$$
	where $\mu_{ik} = \left| \left\{(\delta_{ik}^*)^{-1} - \widehat\delta_{ik}^{-1}\right\} V_{ik}\right| + \left| \widehat\delta_{ik}^{-1}(\omega_{i}-\omega_{k}) \right|$ .
	
	Note that $(\delta_{ik}^*)^{-1}V_{ik} \to N(0,1)$ in distribution following \eqref{eq:rank normality}, then it follows from the Berry–Esseen theorem that
	$$
	\max_{k\in\SSSS_0}\sup_{t\in\R}\big | \Pr\left( (\delta_{ik}^*)^{-1}V_{ik} \leq t \right) - \Phi(t) \big | = O\left(\frac{1}{\sqrt{n}}\right).
	$$ 
   It further implies that
	\begin{equation}\label{eq:fdr 1term}
		\begin{aligned}
			&\max_{k\in\SSSS_0}\Pr\left((\delta_{ik}^*)^{-1} | V_{ik} |
			\geq \Phi^{-1}\left(1-\frac{\alpha \min(l,K)}{2KL}\right) - \frac{e^{20\|\ttt^*\|_{\infty}}(\log n)^2}{\sqrt{n}} \right) \\
			\leq& 2\left\{1 - \Phi\left[  \Phi^{-1}\left(1-\frac{\alpha \min(l,K)}{2KL}\right) - \frac{(\log n)^2}{\sqrt{n}} \right] \right\} + O\left(\frac{1}{\sqrt{n}}\right) \\
			=&\frac{\alpha \min(l,K)}{KL} + O\left(\frac{e^{20\|\ttt^*\|_{\infty}}(\log n)^2}{\sqrt{n}}\right).
		\end{aligned}
	\end{equation}
	
	Also, it follows from \eqref{eq:inverse2} that
	\begin{align*}
	& \max_{k\in\SSSS_0}\Pr\left(\mu_{ik} \geq \frac{(\log n)^2}{\sqrt{n}}\right) \\
	\leq & \max_{k\in\SSSS_0}\Pr\left(\left|(\delta_{ik}^*)^{-1} - \widehat\delta_{ik}^{-1}\right| \left\{ | V_{ik} |+\left|\omega_{i}-\omega_{k}\right|\right\} + 
		\left| (\delta_{ik}^*)^{-1}(\omega_{i}-\omega_{k}) \right| \geq \frac{e^{20\|\ttt^*\|_{\infty}}(\log n)^2}{\sqrt{n}}\right).
	\end{align*}
	But it follows from \eqref{eq:vs residual} and \eqref{eq:inverse2} that 
	 $$
	\max_{k\in\SSSS_0}\left|(\delta_{ik}^*)^{-1} - \widehat\delta_{ik}^{-1}\right| \left\{|V_{ik} |+\left|\omega_{i}-\omega_{k}\right|\right\} + 
	\left| (\delta_{ik}^*)^{-1}(\omega_{i}-\omega_{k}) \right| \lesssim_P \frac{e^{20\|\ttt^*\|_{\infty}}\log n}{\sqrt{n}}.
	$$ 
	Therefore, $\max_{k\in\SSSS_0}\Pr\left(\mu_{ik} \geq  e^{20\|\ttt^*\|_{\infty}}(\log n)^2/\sqrt{n}\right) = O\left(\frac{1}{n}\right)$ following Lemma~\ref{lem:concentrate}.
		 
	Putting the above results together, we obtain that 
	$$
	\begin{aligned}
		\max_{k\in\SSSS_0}\Pr\left(p_k\leq\frac{\alpha \min(l,K)}{KL}\right) &= \frac{\alpha \min(l,K)}{KL} + O\left( \frac{e^{20\|\ttt^*\|_{\infty}}(\log n)^2}{\sqrt{n}} \right) + O \left (\frac{1}{n}\right)\\ 
		&= \frac{\alpha \min(l,K)}{KL} + O\left(\frac{e^{20\|\ttt^*\|_{\infty}}(\log n)^2}{\sqrt{n}}\right).
	\end{aligned}
	$$ 
	It then implies that
	$$
	\begin{aligned}
		\FDR &\leq \sum_{l=1}^{\infty}\frac{1}{l(l+1)}\sum_{k\in\SSSS_0}\Pr\left(p_k\leq\frac{\alpha \min(l,K)}{KL}\right) \\
		&= \sum_{k\in\SSSS_0}\left\{\sum_{l=1}^{K}\frac{1}{l(l+1)}\Pr\left(p_k\leq\frac{\alpha l}{KL}\right)+ \sum_{l=K+1}^{\infty}\frac{1}{l(l+1)}\Pr\left(p_k\leq\frac{\alpha }{L}\right)\right\} \\
		&= \sum_{k\in\SSSS_0}\left\{\sum_{l=1}^{K}\frac{1}{l(l+1)}\left[\frac{\alpha l}{KL}+O\left(\frac{e^{20\|\ttt^*\|_{\infty}}(\log n)^2}{\sqrt{n}}\right)\right]+ \sum_{l=K+1}^{\infty}\frac{1}{l(l+1)}\left[ \frac{\alpha}{L} + O\left(\frac{e^{20\|\ttt^*\|_{\infty}}(\log n)^2}{\sqrt{n}}\right) \right]\right\} \\
		&\leq \frac{\alpha K_0}{K} + \frac{\alpha K_0}{KL} + O\left(\frac{e^{20\|\ttt^*\|_{\infty}}K_0(\log n)^2}{\sqrt{n}}\right).
	\end{aligned}
	$$ 
	As $e^{20\|\ttt^*\|_{\infty}}K_0n^{-1/2}(\log n)^2 = o(1),$ the desired upper bound on $\FDR$ follows immediately.	\end{proof}

\section*{Appendix C: proof of lemmas}

\begin{proof}[\bf Proof of Lemma~\ref{lem:var}] 
First, simple algebra yields that
	\begin{equation}
	\label{eqn:lemma1}
	\begin{aligned}
		u_i &= -\EEE \left[ \frac{\partial^2 l(\ttt^*;\sssp)}{\partial \vs_i^2} \right] = -\sum_{j=1,j\neq i}^n \EEE\left[\frac{\partial^2 l(\ttt^*;\sssp)}{\partial\vs_i\partial\vc_j}\right],~\mbox{for}~i \in [n],\\
		u_{n+j} &= -\EEE \left[ \frac{\partial^2 l(\ttt^*;\sssp)}{\partial \vc_j^2} \right] = -\sum_{i=1,i\neq j}^n \EEE\left[\frac{\partial^2 l(\ttt^*;\sssp)}{\partial\vs_i\partial\vc_j}\right],~\mbox{for}~j \in [n].
	\end{aligned}
	\end{equation}
	It is shown that $\EEE[-\frac{\partial^2 l(\ttt;\sssp)}{\partial \vs_i\partial\vc_j}] = \frac{e^{\vs_i+\vc_j}}{(1+e^{\vs_i+\vc_j})^2}\left[ 2- \frac{(1-\ssp_i)e^{\vs_i+\vc_j}}{e^{\vs_i+\vc_j}+(1-\ssp_i)} - \frac{1-\ssp_i}{e^{\vs_i+\vc_j}+\frac{1-\ssp_i}{1+\ssp_i}} \right]$ for any $i\neq j$. 
Also, we have $\max_{i,j}|\vs_i^*+\vc_j^*| \leq 2\|\ttt^*\|_{\infty}$. It can be verified that there exists a positive constant $c$ such that 
$$
c^{-1}e^{-2\|\ttt^*\|_{\infty}} \le \frac{e^{\vs_i+\vc_j}}{(1+e^{\vs_i+\vc_j})^2}\left[ 2- \frac{(1-\ssp_i)e^{\vs_i+\vc_j}}{e^{\vs_i+\vc_j}+(1-\ssp_i)} - \frac{1-\ssp_i}{e^{\vs_i+\vc_j}+\frac{1-\ssp_i}{1+\ssp_i}} \right] \leq c
$$ for any $1 \le i, j \le n$. It then follows from \eqref{eqn:lemma1} that $c^{-1}ne^{-2\|\ttt^*\|_{\infty}} \le  \min_{1\leq i\leq 2n} u_i\leq \max_{1\leq i\leq 2n} u_i \le cn$. 
Further, we have  $c^{-1}e^{-2\|\ttt^*\|_{\infty}}\leq \frac{(1+\ssp_i)e^{\vs_i^*+\vc_j^*}}{(1+e^{\vs_i^*+\vc_j^*})^2}\leq c$ for any $1 \le i, j \le n$, and then it follows from the definition of $v_i$ that $c^{-1}ne^{-2\|\ttt^*\|_{\infty}} \le  \min_{1\leq i\leq 2n} v_i\leq \max_{1\leq i\leq 2n} v_i \le cn$. Similarly, note that $c^{-1}e^{-2\|\ttt^*\|_{\infty}}\leq \frac{(1+\ssp_i)e^{\vs_i^*+\vc_j^*}+\ssp_i(1-\ssp_i)}{(1+e^{\vs_i^*+\vc_j^*})^2}\leq c$ for any $1 \le i, j \le n$, which implies that $c^{-1}ne^{-2\|\ttt^*\|_{\infty}} \le  \min_{1\leq i\leq 2n} w_i\leq \max_{1\leq i\leq 2n} w_i \le cn$.

Let $\MM = (m_{ij})_{n\times n}$ with $m_{ij} = \vc_j + \vs_i$. As $l(\ttt;\sssp)$ depends on $\alpha_i$ and $\beta_j$ only through $m_{ij}$, we may write $l(\ttt;\sssp)$ as $l(\MM;\sssp)$ with $\MM=(m_{ij})_{i,j=1}^n$ without causing any confusion, and it holds that $\frac{\partial^2 l(\ttt;\sssp)}{\partial \vs_i\partial\vc_j} = \frac{\partial^2 l(\MM;\sssp)}{\partial m_{ij}^2}$. Next, it also follows from \eqref{eqn:lemma1} that
	$$
	\begin{aligned}
		& \max_{1\leq i\leq n} \left|\frac{\partial^2 l(\ttt;\sssp)}{\partial \theta_i^2} - \frac{\partial^2 l(\ttt^*;\sssp)}{\partial \theta_i^2} \right| \leq \max_{1\leq i\leq n}\sum_{j=1,j\neq i}^n \left| \frac{\partial^2 l(\ttt;\sssp)}{\partial \vs_i\partial\vc_j} - \frac{\partial^2 l(\ttt^*;\sssp)}{\partial \vs_i\partial\vc_j} \right| \\
		= & \max_{1\leq i\leq n}\sum_{j=1,j\neq i}^n \left| \frac{\partial^2 l(\MM;\sssp)}{\partial m_{ij}^2} - \frac{\partial^2 l(\MM^*;\sssp)}{\partial m_{ij}^2} \right| \lesssim \sum_{j=1,j\neq i}^n |m_{ij} - m_{ij}^*| \lesssim n\|\ttt-\ttt^*\|_{\infty},
	\end{aligned}
	$$ 
	where the second inequality is due to the fact that there exists a constant $\epsilon>0$ such that $$
	\sup_{\ttt:\|\ttt-\ttt^*\|_{\infty}\leq \epsilon}\max_{i\neq j} \left| \frac{\partial^3 l(\MM;\sssp)}{\partial m_{ij}^3} \right| = O(1).
	$$
	Similarly, we also have 
	$$
	\max_{n+1\leq i\leq 2n} \left| \frac{\partial^2 l(\ttt;\sssp)}{\partial \theta_i^2} - \frac{\partial^2 l(\ttt^*;\sssp)}{\partial \theta_i^2} \right| \lesssim n\|\ttt-\ttt^*\|_{\infty}.
	$$ 
	The proof is similarly for $$
n^{-1} \max_{1\leq i\leq 2n} \left| \frac{\partial^2 l(\ttt;\sssp)}{\partial \theta_i^2} - \frac{\partial^2 l(\ttt;\sssp^*)}{\partial \theta_i^2} \right| = O(\|\sssp-\sssp^*\|_{\infty}),
	$$ which completes the proof of Lemma \ref{lem:var}.
\end{proof}


\begin{proof}[\bf Proof of Lemma \ref{lem:concentrate}]
	Note that $|y_{ij}|\leq 1$ for any $i,j \in [n]$, and it can also be verified that $\EEE\frac{\partial l(\ttt^*;\sssp^*)}{\partial\theta_i} = 0$ and $n^{-1} \Big |\frac{\partial l(\ttt^*;\sssp^*)}{\partial\theta_i} \Big |\leq 2$ for any $i \in [2n]$. It then follows from Hoeffding's inequality that
	$$
	\begin{aligned}
		\Pr\left(\max_{1\leq i\leq 2n}|g_i - \EEE g_i| \geq \sqrt{4(n-1)\log (n-1)}\right) & \leq \sum_{i=1}^{2n}\Pr\left(|g_i - \EEE g_i| \geq \sqrt{4(n-1)\log (n-1)}\right) \\
		&\leq 4n\exp\left( -\frac{8(n-1)\log(n-1)}{4(n-1)} \right)  = \frac{4n}{(n-1)^2},
	\end{aligned}
	$$ and $$
	\begin{aligned}
		\Pr\left(\max_{1\leq i\leq 2n}\left|\frac{\partial l(\ttt^*;\sssp^*)}{\partial\theta_i}\right| \geq \sqrt{16(n-1)\log (n-1)}\right) &\leq \sum_{i=1}^{2n}\Pr\left(\left|\frac{\partial l(\ttt^*;\sssp^*)}{\partial\theta_i}\right| \geq \sqrt{16(n-1)\log (n-1)}\right)\\
		&\leq 4n\exp\left( -\frac{32(n-1)\log(n-1)}{16(n-1)} \right) = \frac{4n}{(n-1)^2}.
	\end{aligned}
	$$ Then, \eqref{eq:con2} holds since by \eqref{lem:var}, we have $$
	\begin{aligned}
	\left|\max_{1\leq i\leq 2n}\left|\frac{\partial l(\ttt^*;\sssp)}{\partial\theta_i}\right| - \max_{1\leq i\leq 2n}\left|\frac{\partial l(\ttt^*;\sssp^*)}{\partial\theta_i}\right|\right| &\leq \max_{1\leq i\leq 2n}\left|\frac{\partial l(\ttt^*;\sssp)}{\partial\theta_i} - \frac{\partial l(\ttt^*;\sssp^*)}{\partial\theta_i}\right| \lesssim n\|\sssp-\sssp^*\|_{\infty}.
	\end{aligned}
	$$
	
Also, there exists a positive constant $c$ such that $|\III_{ij}| \leq c/2$ for any $i,j \in [n]$. Again, by Hoeffding's inequality, we have
 $$
	\begin{aligned}
		& \Pr\left(\max_{1\leq i\leq 2n} \left |\frac{\partial^2 l(\ttt^*;\sssp)}{\partial \theta_i^2}+u_i \right | \geq c\sqrt{(n-1)\log (n-1)}\right) \\
		\leq & \sum_{i=1}^{2n}\Pr\left( \left|\frac{\partial^2 l(\ttt^*;\sssp)}{\partial \theta_i^2}+u_i \right | \geq c\sqrt{(n-1)\log (n-1)}\right)\\
		\leq & 4n\exp\left( -\frac{2c^2(n-1)\log(n-1)}{c^2(n-1)} \right) = \frac{4n}{(n-1)^2}.
	\end{aligned}
	$$
This completes the proof of Lemma \ref{lem:concentrate}.
\end{proof}

\begin{proof}[\bf Proof of Lemma~\ref{lem:initial consis}]
	
	The proof mainly follows from the results in \cite{yan2016asymptotics}, except that we have $\sssp$ here which approximates true values $\sssp^*$.
	Note that $$
	\begin{aligned}
		\frac{\partial F_i(\ttt^*;\sssp)}{\partial\vs_i} &= -\sum_{k=1,k\neq i}^n \frac{(1+\ssp_i)e^{\vs_i^*+\vc_k^*}}{(1+e^{\vs_i^*+\vc_k^*})^2}= -v_i,&&~\mbox{for}~i \in [n],\\
		\frac{\partial F_{n+j}(\ttt^*;\sssp)}{\partial\vc_j} &= -\sum_{k=1,k\neq j}^n \frac{(1+\ssp_k)e^{\vs_k^*+\vc_j^*}}{(1+e^{\vs_k^*+\vc_j^*})^2}=-v_{n+j},&&~\mbox{for}~j \in [n-1].
	\end{aligned}
	$$ By Lemma~\ref{lem:var},
	there exist positive constants $c_1$ and $c_2$ such that $-\partial F(\ttt;\sssp)/\partial \ttt \in\mathcal L(c_1e^{-2\|\ttt^*\|_{\infty}},c_2),$ where $\mathcal L(c_1e^{-2\|\ttt^*\|_{\infty}},c_2)$ is defined as in Section 2.1 of \cite{yan2016asymptotics}. 
	
	By Lemma~\ref{lem:concentrate}, with probability at least $1-4n/(n-1)^2$, we have $\max\big\{ \max_{1\leq i\leq n}|d_i - \EEE d_i|, \max_{1\leq j\leq n}|b_j-\EEE b_j| \big\} \leq \sqrt{4(n-1)\log (n-1)}$. Then, $$
	\begin{aligned}
	\max_{1\leq i\leq n}|F_i(\ttt;\sssp)| \leq& \max_{1\leq i\leq n}|g_i-\EEE g_i| + \max_{1\leq i\leq n}\sum_{k=1,k\neq i}^n\frac{|\ssp_i-\ssp_i^*|}{1+e^{\vs_i^*+\vc_k^*}}\\ 
	\leq& \sqrt{4(n-1)\log (n-1)} + (n-1)\|\sssp-\sssp^*\|_{\infty} \\
	=&[1+o(1)]\sqrt{4(n-1)\log (n-1)}.
	\end{aligned}
	$$ Similarly, $\max_{1\leq j\leq n}|F_{n+j}(\ttt;\sssp)|\leq [1+o(1)]\sqrt{4(n-1)\log (n-1)}$.
	According to Theorem 7 in \cite{yan2016asymptotics} with 
	$m = c_1e^{-2\|\ttt^*\|_{\infty}}$, $M = c_2$, $K_1 = (1+\ssp)(n-1)$, $K_2 = \frac{(1+\ssp)(n-1)}{2}$, $\rho \asymp e^{6\|\ttt^*\|_{\infty}}$
	and $r \lesssim e^{6\|\ttt^*\|_{\infty}}\sqrt{\log n/ n}$, it holds that with probability approaching 1, $F(\ttt;\sssp) = 0$ has a solution $\check\ttt$, and it satisfies $$
	\|\check\ttt-\ttt^*\|_{\infty}\leq \frac{r}{1-\rho r} \leq 2r\lesssim e^{6\|\ttt^*\|_{\infty}}\sqrt{\frac{\log n}{n}}.
	$$
	
	Next, we prove the uniqueness of $\check \ttt$ by contradiction. If there exists $\widetilde \ttt \neq \check\ttt$ such that $F(\widetilde\ttt;\sssp) = \bf0$, we define $h(t) = (t \widetilde \ttt + (1-t)\check\ttt)^\top F(t \widetilde \ttt + (1-t)\check\ttt;\sssp)$, and 
	$$
	h'(t) = (\widetilde \ttt - \check\ttt)^\top \left[\frac{\partial F(t\widetilde \ttt + (1-t)\check\ttt;\sssp)}{\partial\ttt}\right](\widetilde \ttt - \check\ttt).
	$$
	Since $-\partial F(\ttt;\sssp)/\partial\ttt$ is a diagonally dominant matrix with positive diagonals \citep{yan2016asymptotics}, it implies that $-\partial F(\ttt;\sssp)/\partial\ttt$ is a positive definite matrix, and thus $h'(t) < 0$ for any $t \in[0,1]$. This contradicts with the fact that $h(0) = h(1) = 0$, and the uniqueness of $\check \ttt$ then follows immediately.	
\end{proof}

\begin{proof}[\bf Proof of Lemma~\ref{lem:initial decomp}]
	Let $	\hh = -F(\check\ttt;\sssp) - F'(\ttt^*;\sssp)(\check\ttt-\ttt^*)$ and $\g = (g_1,...,g_{2n-1})^\top$.
	Further denote $$
	\g(\ttt;\sssp) = (g_1(\ttt;\sssp),...,g_{2n-1}(\ttt;\sssp))^\top,
	$$ where $g_i(\ttt;\sssp) = \sum_{k=1,k\neq i}^n\frac{e^{\vs_i+\vc_k} - \ssp_i}{1+^{\vs_i+\vc_k}}$ and $g_{n+j}(\ttt;\sssp) = \sum_{k=1,k\neq j}^n\frac{e^{\vs_k+\vc_j} - \ssp_k}{1+^{\vs_k+\vc_j}}$ for $i,j\in[n]$. Note that $\g(\ttt^*;\sssp^*) = \EEE \g$. Then, $$
	\check\ttt-\ttt^* = -[F'(\ttt^*;\sssp)]^{-1}(\g-\EEE\g) - [F'(\ttt^*;\sssp)]^{-1}(\g(\ttt^*;\sssp^*) - \g(\ttt^*;\sssp)) - [F'(\ttt^*;\sssp)]^{-1}\hh.
	$$ By Lemma \ref{lem:initial consis}, and Lemmas 8, 9 of \cite{yan2016asymptotics}, we have $\|[F'(\ttt^*;\sssp)]^{-1}\hh\|_{\infty} \lesssim e^{6\|\ttt^*\|_{\infty}}\|\check\ttt-\ttt^*\|_{\infty}^2\lesssim_Pe^{18\|\ttt^*\|_{\infty}}\log n/n$. Similarly, we have $\|[F'(\ttt^*;\sssp)]^{-1}(\g(\ttt^*;\sssp^*) - \g(\ttt^*;\sssp))\|_{\infty}\lesssim e^{6\|\ttt^*\|_{\infty}}\|\sssp-\sssp^*\|_{\infty}\lesssim e^{18\|\ttt^*\|_{\infty}}\log n/n$.
	The desired results follows immediately.
\end{proof}


\begin{proof}[\bf Proof of Lemma~\ref{lem:multi concentrate12}]
Note that for any $l \in [n-1],$ $$
	u_i^{-1}\EEE\left[\frac{\partial^2l(\ttt^*;\sssp)}{\partial\vs_i\partial\vc_l}\right] \sqrt{n} v_{n+l}^{-1}(g_{n+l}-\EEE g_{n+l})
	$$ 
is a sub-Gaussian random variable with zero mean and by Lemma~\ref{lem:var}, 
$$
	\max_{1\leq l\leq n-1} \var\left\{ u_i^{-1}\EEE\left[\frac{\partial^2l(\ttt^*;\sssp)}{\partial\vs_i\partial\vc_l}\right] \sqrt{n} v_{n+l}^{-1}(g_{n+l}-\EEE g_{n+l}) \right\} \leq \frac{c_1e^{8\|\ttt^*\|_{\infty}}}{n^2},
	$$ 
where $c_1$ is a constant. We have 
$$
	\begin{aligned}
		&\Pr\left(\max_{1\leq i\leq n}\left| \sum_{l=1}^{n-1}  u_i^{-1}\EEE\left[\frac{\partial^2l(\ttt^*;\sssp)}{\partial\vs_i\partial\vc_l}\right] \sqrt{n} v_{n+l}^{-1}(g_{n+l}-\EEE g_{n+l})  \right|
		\leq c_2e^{4\|\ttt^*\|_{\infty}}\sqrt{\frac{\log n}{n}}\right)\\ 
		\geq& 1-n\max_{1\leq i\leq n}\Pr\left(\left| \sum_{l=1}^{n-1}  u_i^{-1}\EEE\left[\frac{\partial^2l(\ttt^*;\sssp)}{\partial\vs_i\partial\vc_l}\right] \sqrt{n} v_{n+l}^{-1}(g_{n+l}-\EEE g_{n+l})  \right|
		\leq c_2e^{4\|\ttt^*\|_{\infty}}\sqrt{\frac{\log n}{n}}\right)  \\
		\geq& 1 - 2n\exp\left(-\frac{c_2^2\log n}{2c_1}\right)= 1-\frac{2}{n},
	\end{aligned}
	$$ 
	where the second inequality is due to Hoeffding's inequality, and
	the last equality holds with $c_2 = 2\sqrt{c_1}.$ Similarly, we can establish the bound for
	$$
	\max_{1\leq j\leq n}\left| \sum_{k=1}^n  u_{n+j}^{-1}\EEE\left[\frac{\partial^2 l(\ttt^*;\sssp)}{\partial\vs_k\partial\vc_j}\right] \sqrt{n} v_k^{-1}(g_k-\EEE g_k) \right|,
	$$ 
	and thus we complete the proof of \eqref{eq:multi concentrate1}.
	
	For \eqref{eq:multi concentrate2}, note that for any $l \in [n-1],$ 
	$$
u_i^{-1}\III_{il}\sqrt{n} v_{n+l}^{-1}(g_{n+l\backslash i}-\EEE g_{n+l\backslash i})  
$$ 
is also a sub-Gaussian random variable with zero mean and by Lemma~\ref{lem:var}, 
$$
\max_{1\leq l\leq n-1} \var\left\{ u_i^{-1}\III_{il}\sqrt{n} v_{n+l}^{-1}(g_{n+l\backslash i}-\EEE g_{n+l\backslash i})   \right\} \leq\frac{c_1e^{8\|\ttt^*\|_{\infty}}}{n^2},
$$ 
where $c_1$ is a constant. The rest of the proof is similar to that of \eqref{eq:multi concentrate1}. \end{proof}

\begin{proof}[\bf Proof of Lemma~\ref{lem:multi cent}]
By Lemma~\ref{lem:initial decomp}, we have $$
\begin{aligned}
&\sum_{k=1}^n  u_{2n}^{-1}\EEE\left[\frac{\partial^2 l(\ttt^*;\sssp)}{\partial\vs_k\partial\vc_n}\right] \sqrt{n}(\check\vs_k - \vs_k^*) + \sum_{l=1}^{n-1}  u_i^{-1}\EEE\left[\frac{\partial^2l(\ttt^*;\sssp)}{\partial\vs_i\partial\vc_l}\right] \sqrt{n}(\check\vc_l - \vc_l^*) \\
=& \sum_{k=1}^n  u_{2n}^{-1}\EEE\left[\frac{\partial^2 l(\ttt^*;\sssp)}{\partial\vs_k\partial\vc_n}\right] \sqrt{n}\left[ v_k^{-1}(g_k-\EEE g_k)+ v_{2n}^{-1}(g_{2n} - \EEE g_{2n}) + \epsilon_i\right] \\
&+ \sum_{l=1}^{n-1}  u_i^{-1}\EEE\left[\frac{\partial^2l(\ttt^*;\sssp)}{\partial\vs_i\partial\vc_l}\right] \sqrt{n}\left[ v_{n+l}^{-1}(g_{n+l}-\EEE g_{n+l}) - v_{2n}^{-1}(b_{n}-\EEE b_{n}) + \epsilon_{n+l} \right] \\
=&\sum_{k=1}^n  u_{2n}^{-1}\EEE\left[\frac{\partial^2 l(\ttt^*;\sssp)}{\partial\vs_k\partial\vc_n}\right] \sqrt{n} v_k^{-1}(g_k-\EEE g_k) + \sum_{l=1}^{n-1}  u_i^{-1}\EEE\left[\frac{\partial^2l(\ttt^*;\sssp)}{\partial\vs_i\partial\vc_l}\right] \sqrt{n} v_{n+l}^{-1}(g_{n+l}-\EEE g_{n+l}) \\
&+ u_i^{-1}\EEE\left[\frac{\partial^2l(\ttt^*;\sssp)}{\partial\vs_i\partial\vc_n}\right]\sqrt{n} v_{2n}^{-1}(g_{2n}-\EEE g_{2n}) + \sum_{k=1}^n  u_{2n}^{-1}\EEE\left[\frac{\partial^2 l(\ttt^*;\sssp)}{\partial\vs_k\partial\vc_n}\right] \sqrt{n}\epsilon_i \\
&+ \sum_{l=1}^{n-1}  u_i^{-1}\EEE\left[\frac{\partial^2l(\ttt^*;\sssp)}{\partial\vs_i\partial\vc_l}\right] \sqrt{n} \epsilon_{n+l}\\
=:& \sum_{k=1}^n  u_{2n}^{-1}\EEE\left[\frac{\partial^2 l(\ttt^*;\sssp)}{\partial\vs_k\partial\vc_n}\right] \sqrt{n} v_k^{-1}(g_k-\EEE g_k) + \sum_{l=1}^{n-1}  u_i^{-1}\EEE\left[\frac{\partial^2l(\ttt^*;\sssp)}{\partial\vs_i\partial\vc_l}\right] \sqrt{n} v_{n+l}^{-1}(g_{n+l}-\EEE g_{n+l}) + r_i.
\end{aligned}
$$ 
By \eqref{eq:multi concentrate1} in Lemma~\ref{lem:multi concentrate12}, 
we have 
\begin{align*}
\left|\sum_{k=1}^n  u_{2n}^{-1}\EEE\left[\frac{\partial^2 l(\ttt^*;\sssp)}{\partial\vs_k\partial\vc_n}\right] \sqrt{n} v_k^{-1}(g_k-\EEE g_k)\right| & \lesssim_P e^{4\|\ttt^*\|_{\infty}}\sqrt{\frac{\log n}{n}}; \\
\max_{1\leq i\leq n}\left|\sum_{l=1}^{n-1}  u_i^{-1}\EEE\left[\frac{\partial^2l(\ttt^*;\sssp)}{\partial\vs_i\partial\vc_l}\right] \sqrt{n} v_{n+l}^{-1}(g_{n+l}-\EEE g_{n+l})\right| & \lesssim_P e^{4\|\ttt^*\|_{\infty}}\sqrt{\frac{\log n}{n}},
\end{align*}
and it follows from Lemmas~\ref{lem:var}, \ref{lem:concentrate} and \ref{lem:initial decomp}
that for any $i \in [n]$, $$
|r_i| \lesssim_P e^{20\|\ttt^*\|_{\infty}}\log n/\sqrt{n}+ e^{8\|\ttt^*\|_{\infty}}\sqrt{n}\|\sssp-\sssp^*\|_{\infty}\lesssim e^{20\|\ttt^*\|_{\infty}}\log n/\sqrt{n}.
$$ 
Therefore, 
it holds true that 
{\footnotesize $$
\max_{1\leq i\leq n}\left|\sum_{k=1}^n  u_{2n}^{-1}\EEE\left[\frac{\partial^2 l(\ttt^*;\sssp)}{\partial\vs_k\partial\vc_n}\right] \sqrt{n}(\check\vs_k - \vs_k^*) + \sum_{l=1}^{n-1}  u_i^{-1}\EEE\left[\frac{\partial^2l(\ttt^*;\sssp)}{\partial\vs_i\partial\vc_l}\right] \sqrt{n}(\check\vc_l - \vc_l^*)\right| \lesssim_P \frac{e^{20\|\ttt^*\|_{\infty}}\log n}{\sqrt{n}}.
$$ }
Similarly, we can also show that   
{\footnotesize $$
\max_{1\leq j\leq n-1}\left|\sum_{k=1}^n \sqrt{n} (\check\vs_k - \vs_k^*) \left\{ u_{n+j}^{-1} \EEE\left[\frac{\partial^2 l(\ttt^*;\sssp)}{\partial\vs_k\partial\vc_j}\right] - u_{2n}^{-1} \EEE\left[\frac{\partial^2 l(\ttt^*;\sssp)}{\partial\vs_k\partial\vc_n}\right] \right\} \right| \lesssim_P \frac{e^{20\|\ttt^*\|_{\infty}}\log n}{\sqrt{n}}.
$$} 
This completes the proof of Lemma  \ref{lem:multi cent}. \end{proof}



\begin{proof}[\bf Proof of Lemma~\ref{lem:multi cent2}]
By Lemma~\ref{lem:initial decomp}, we have $$
\begin{aligned}
&\sum_{k=1}^n  u_{2n}^{-1}\III_{kn}\sqrt{n}(\check\vs_k - \vs_k^*) + \sum_{l=1}^{n-1}  u_i^{-1}\III_{il} \sqrt{n}(\check\vc_l - \vc_l^*) \\
=& \sum_{k=1}^n  u_{2n}^{-1}\III_{kn} \sqrt{n}\left[ v_k^{-1}(g_k-\EEE g_k)+v_{2n}^{-1}(g_{2n} - \EEE g_{2n}) + \epsilon_i\right] \\
&+ \sum_{l=1}^{n-1}  u_i^{-1}\III_{il} \sqrt{n}\left[ v_{n+l}^{-1}(g_{n+l}-\EEE g_{n+l}) - v_{2n}^{-1}(b_{n}-\EEE b_{n}) + \epsilon_{n+l} \right] \\
=&\sum_{k=1}^n  u_{2n}^{-1}\III_{kn}\sqrt{n} v_k^{-1}(g_k-\EEE g_k) + \sum_{l=1}^{n-1}  u_i^{-1}\III_{il}\sqrt{n} v_{n+l}^{-1}(g_{n+l}-\EEE g_{n+l}) \\
&+u_{2n}^{-1}\sqrt{n} v_{2n}^{-1}(g_{2n}-\EEE g_{2n}) \sum_{k=1}^n\III_{kn} - u_i^{-1}\sqrt{n} v_{2n}^{-1}(b_{n}-\EEE b_{n})\sum_{l=1}^{n-1}\III_{il} \\
&+ \sum_{k=1}^n  u_{2n}^{-1}\III_{kn} \sqrt{n}\epsilon_i + \sum_{l=1}^{n-1}  u_i^{-1}\III_{il} \sqrt{n} \epsilon_{n+l}\\
=&\sum_{k=1}^n  u_{2n}^{-1}\III_{kn}\sqrt{n} v_k^{-1}(g_{k\backslash n}-\EEE g_{k\backslash n}) + \sum_{k=1}^n  u_{2n}^{-1}\III_{kn}\sqrt{n} v_k^{-1}(y_{kn}-\EEE y_{kn}) \\
&+ \sum_{l=1}^{n-1}  u_i^{-1}\III_{il}\sqrt{n} v_{n+l}^{-1}(g_{n+l\backslash i}-\EEE g_{n+l\backslash i}) + \sum_{l=1}^{n-1}  u_i^{-1}\III_{il}\sqrt{n} v_{n+l}^{-1}(y_{il}-\EEE y_{il}) \\
&+\sqrt{n} v_{2n}^{-1}(g_{2n}-\EEE g_{2n})\frac{u_{2n}-u_{2n}}{u_{2n}} - \sqrt{n} v_{2n}^{-1}(b_{n}-\EEE b_{n})\frac{u_i-u_i}{u_i} + \sqrt{n} v_{2n}^{-1}(b_{n}-\EEE b_{n})\frac{\III_{in}}{u_i} \\
&+ \sum_{k=1}^n  u_{2n}^{-1}\III_{kn} \sqrt{n}\epsilon_i + \sum_{l=1}^{n-1}  u_i^{-1}\III_{il} \sqrt{n} \epsilon_{n+l}\\
=:&\sum_{k=1}^n  u_{2n}^{-1}\III_{kn}\sqrt{n} v_k^{-1}(g_{k\backslash n}-\EEE g_{k\backslash n}) + \sum_{l=1}^{n-1}  u_i^{-1}\III_{il}\sqrt{n} v_{n+l}^{-1}(g_{n+l\backslash i}-\EEE g_{n+l\backslash i}) + s_i.
\end{aligned}
$$
By \eqref{eq:multi concentrate2} in Lemma~\ref{lem:multi concentrate12}, 
we have 
\begin{align*}
\left|\sum_{k=1}^n  u_{2n}^{-1}\III_{kn}\sqrt{n} v_k^{-1}(g_{k\backslash n}-\EEE g_{k\backslash n}) \right| & \lesssim_P e^{4\|\ttt^*\|_{\infty}}\sqrt{\frac{\log n}{n}}; \\
\max_{1\leq i\leq n}\left| \sum_{l=1}^{n-1}  u_i^{-1}\III_{il}\sqrt{n} v_{n+l}^{-1}(g_{n+l\backslash i}-\EEE g_{n+l\backslash i})  \right| & \lesssim_P e^{4\|\ttt^*\|_{\infty}}\sqrt{\frac{\log n}{n}},
\end{align*}
and it follows from Lemmas~\ref{lem:var}, \ref{lem:concentrate} and \ref{lem:initial decomp}
that $$
\max_{1\leq i\leq n} |s_i| \lesssim_P e^{20\|\ttt^*\|_{\infty}}\log n/\sqrt{n} + e^{8\|\ttt^*\|_{\infty}}\sqrt{n}\|\sssp-\sssp^*\|_{\infty}\lesssim e^{20\|\ttt^*\|_{\infty}}\log n/\sqrt{n}.
$$

Therefore, it holds true that 
$$
\begin{aligned}
\max_{1\leq i\leq n}\left|\sum_{k=1}^n  u_{2n}^{-1}\III_{kn}\sqrt{n}(\check\vs_k - \vs_k^*) + \sum_{l=1}^{n-1}  u_i^{-1}\III_{il} \sqrt{n}(\check\vc_l - \vc_l^*)\right| \lesssim_P \frac{e^{20\|\ttt^*\|_{\infty}}\log n}{\sqrt{n}}.
\end{aligned}
$$ 
Similarly, we can also show that 
$$
\max_{1\leq j\leq n-1}\left|\sum_{k=1}^n \sqrt{n} (\check\vs_k - \vs_k^*) \left\{ u_{n+j}^{-1} \III_{kj} - u_{2n}^{-1} \III_{kn}\right\} \right| \lesssim_P \frac{e^{20\|\ttt^*\|_{\infty}}\log n}{\sqrt{n}}.
$$ 
This completes the proof of Lemma  \ref{lem:multi cent2}.  \end{proof}

\section*{Appendix D: estimation of $\sssp$ and asymptotics}

 Define $\III_1^* = \{i\in[n]:\ssp_i = \ssp_{01}\}$ and $n_1^* = |\III_1^*|$. 
We develop a selection procedure to determine the pattern of each node in sending negative edges. 
Specifically, for each $i\in[n]$, we calculate $\zeta_i = n^{-1}\sum_{k=1,k\neq i}^n1_{\{y_{ik}=-1\}}$, and define $\III_1 = \{i\in[n]:\zeta_i> \xi_n\}$, where $\xi_n$ is a pre-specified threshold. Further, let $\ssp_i = \ssp$ if $i\in\III_1$, where $\ssp$ is the unknown parameter we are going to estimate, and $\ssp_i = \log n / n$ if otherwise. 
Without loss of generality, we assume $\III_1 = [{n_1}]$ and define $\ttt_{1} = (\vs_1,...,\vs_{n_1},\vc_1,...,\vc_{{n_1}-1})^\top$. 
Then, the log likelihood based on the subnetwork with $n_1$ nodes takes the form $$
l(\ttt_{1};\ssp) = \sum_{i,j=1,i\neq j}^{n_1} l_{ij}(\vs_i+\vc_j;\ssp).
$$ 
To estimate $\ssp$, we consider the following restricted MLE $$
(\widetilde\ttt_{1},\widetilde\ssp) = \argmin_{\substack{\|\ttt_{1}\|_{\infty}\leq \tau\log n_1 \\ \ssp\in[\gamma_n,1-\gamma_n]}}l(\ttt_{1};\ssp) ,
$$ where $\tau\geq 1/40$ and $\gamma_n\in(0,1/2)$. Given $\widetilde\ssp$, we let $\widetilde\sssp = (\widetilde\ssp_1,...,\widetilde\ssp_n)$, where $\widetilde\ssp_i = \widetilde\ssp$ if $i\in\III_1$ and $\widetilde\ssp_i = \log n / n$ if otherwise.

\begin{lemma}\label{lem:selection}
Suppose $\|\ttt^*\|_{\infty} \leq c\log n$ with $0<c<1/40$, $\ssp_{00}\lesssim e^{12\|\ttt^*\|_{\infty}}\log n/n$ and $\ssp_{01}\in(\gamma_n,1-\gamma_n)$.
Choose $\xi_n$ such that $\sqrt{\log n/ n}\ll \xi_n \ll \gamma_nn^{-\frac{1}{10}}$. Then, $\Pr(\III_1= \III_1^*)\to 1$ as $n$ grows to infinity.
\end{lemma}

\begin{proof}[{\bf Proof of Lemma~\ref{lem:selection}}]
By Hoeffding's inequality, $$
\max_{i\in[n]}|\zeta_i - \EEE \zeta_i| \lesssim_P \sqrt{\frac{\log n}{n}},
$$ which implies that $$
\begin{aligned}
\min_{i\in\III_1^*}|\zeta_i| &\gtrsim_P \min_{i\in\III_1^*}|\EEE\zeta_i| - \sqrt{\frac{\log n}{n}}\gtrsim \gamma_ne^{-4\|\ttt^*\|_{\infty}} - \sqrt{\frac{\log n}{n}}\gtrsim \gamma_nn^{-\frac{1}{10}},\\
\max_{i\in[n]\backslash\III_1^*}|\zeta_i| &\lesssim_P \max_{i\in[n]\backslash\III_1^*}|\EEE\zeta_i| + \sqrt{\frac{\log n}{n}} \lesssim \frac{e^{12\|\ttt^*\|_{\infty}}\log n}{n} + \sqrt{\frac{\log n}{n}}\lesssim \sqrt{\frac{\log n}{n}}.
\end{aligned}
$$ This completes the proof.
\end{proof}

Let $p_{ij}(\vs_i,\vc_j,\ssp)$ represent the distribution of $y_{ij}$ under parameters $(\vs_i,\vc_j,\ssp)$ and $p_{ij} = p_{ij}(\vs_i^*,\vc_j^*,\ssp_{01})$. Define the KL-divergence of $p_{ij}$ from $p_{ij}(\vs_i,\vc_j,\ssp)$ as $$
D_{KL}(p_{ij} || p_{ij}(\vs_i,\vc_j,\ssp)) = \sum_{y\in\{-1,0,1\}}p(y\mid \vs_i^*+\vc_j^*,\ssp_{01})\log\frac{p(y\mid \vs_i^*+\vc_j^*,\ssp_{01})}{p(y\mid \vs_i+\vc_j,\ssp)}.
$$ For any $\rho>0$, define $B_n(\rho) = \{\ssp:|\ssp-\ssp_{01}|<\rho\}$ and $B_n^c(\rho) = [\gamma_n,1-\gamma_n]\backslash B_n(\rho)$.

\begin{lemma}\label{lem:consis ssp}
Under the same conditions of Lemma~\ref{lem:selection}, suppose
\begin{equation}\label{eq:KL}
\min_{\substack{\|\ttt_{1}\|_{\infty}\leq \tau\log n_1 \\ \ssp\in B_n^c(\rho)}} \frac{1}{n_1(n_1-1)}\sum_{i,j=1,i\neq j}^{n_1}D_{KL}(p_{ij}||p_{ij}(\vs_i,\vc_j,\ssp)) \gtrsim \rho e^{-12\|\ttt_1^*\|_{\infty}},
\end{equation}
then we have $|\widetilde\ssp - \ssp_{01}|\lesssim_P e^{12\|\ttt_1^*\|_{\infty}}\log n_1/n_1$.
\end{lemma}

\begin{proof}[\bf Proof of Lemma~\ref{lem:consis ssp}]
The proof is similar to the proof of Theorem 2 in \cite{yan2019stat}, and thus we only show the main steps and different parts. We have $$
\begin{aligned}
&l(\ttt_{1};\ssp) - \EEE l(\ttt_{1};\ssp) = \sum_{i\neq j}^{n_1} \{1_{\{y_{ij}=1\}} - \EEE 1_{\{y_{ij}=1\}}\} \log p(1\mid \vs_i+\vc_j,\ssp)\\
&+\{1_{\{y_{ij}=0\}} - \EEE 1_{\{y_{ij}=0\}}\} \log p(0\mid \vs_i+\vc_j,\ssp)
+\{1_{\{y_{ij}=-1\}} - \EEE 1_{\{y_{ij}=-1\}}\} \log p(-1\mid \vs_i+\vc_j,\ssp).
\end{aligned}
$$ Note that there exists a constant $c$ such that $$
\max_{\substack{\|\ttt_{1}\|_{\infty}\leq \tau\log n_1 \\ \ssp\in[\gamma_n,1-\gamma_n]}}\max_{\substack{y\in\{-1,0,1\} \\ i\neq j}} |\log p(y\mid \vs_i+\vc_j,\ssp)| \leq c\log n_1.
$$ Then, by Hoeffding's equality, there exists a possibly different constant $c_1$ such that
\begin{equation}\label{eq:hoeffding ssp}
\max_{\substack{\|\ttt_{1}\|_{\infty}\leq \tau\log n_1 \\ \ssp\in[\gamma_n,1-\gamma_n]}}\frac{1}{{n_1}({n_1}-1)} |l(\ttt_{1};\ssp) - \EEE l(\ttt_{1};\ssp)| <_P c_1\log n_1\sqrt{\frac{\log n_1(n_1-1)}{n_1(n_1-1)}}<c_1\frac{\log n_1}{n_1},
\end{equation}
where $c_1$ may take different values. 
For any $\rho>0$, define $$
\epsilon_n(\rho) = \frac{1}{{n_1}({n_1}-1)}\left\{ \max_{\|\ttt_{1}\|_{\infty}\leq \tau\log n_1} \EEE[l(\ttt_{1};\ssp_{01})] - \max_{\substack{\|\ttt_{1}\|_{\infty}\leq \tau\log n_1 \\ \ssp\in B_n^c(\rho)}} \EEE[l(\ttt_{1};\ssp)] \right\},
$$ and $\rho_n = \inf\left\{\rho>0:\epsilon_n(\rho)>2c_1\log n_1/n_1\right\}$. Similar as proof of Theorem 2 in \cite{yan2019stat}, \eqref{eq:hoeffding ssp} implies $$
\max_{\|\ttt_{1}\|_{\infty}\leq \tau\log n_1}\frac{1}{{n_1}({n_1}-1)}\EEE[l(\ttt_{1};\widetilde\ssp)] > \max_{\substack{\|\ttt_{1}\|_{\infty}\leq \tau\log n_1 \\ \ssp\in B_n^c(\rho_n)}}\frac{1}{{n_1}({n_1}-1)}\EEE[l(\ttt_{1};\ssp)],
$$ which further leads that $\widetilde\ssp\in B_n(\rho_n)$.

Note that $$
\EEE[l(\ttt_1,\ssp)] = -\sum_{i,j=1,i\neq j}^{n_1} D_{KL}(p_{ij} || p_{ij}(\vs_i,\vc_j,\ssp)) + \sum_{i,j=1,i\neq j}^{n_1} S(p_{ij}),
$$ where $D_{KL}(p_{ij} || p_{ij}(\vs_i,\vc_j,\ssp))$ is the KL-divergence of $p_{ij}$ from $p_{ij}(\vs_i,\vc_j,\ssp)$ as defined earlier, and $$
S(p_{ij}) =\sum_{y\in\{-1,0,1\}}p(y\mid \vs_i^*+\vc_j^*,\ssp_{01})\log p(y\mid \vs_i^*+\vc_j^*,\ssp_{01}).
$$ 
Then by \eqref{eq:KL}, we have $$
\begin{aligned}
&\max_{\|\ttt_{1}\|_{\infty}\leq \tau\log n_1} \EEE[l(\ttt_{1};\ssp_{01})] - \max_{\substack{\|\ttt_{1}\|_{\infty}\leq \tau\log n_1 \\ \ssp\in B_n^c(\rho)}} \EEE[l(\ttt_{1};\ssp)] \\
=& -\min_{\|\ttt_{1}\|_{\infty}\leq \tau\log n_1}\sum_{i,j=1,i\neq j}^{n_1} D_{KL}(p_{ij}||p_{ij}(\vs_i,\vc_j,\ssp_{01})) + \min_{\substack{\|\ttt_{1}\|_{\infty}\leq \tau\log n_1 \\ \ssp\in B_n^c(\rho)}}\sum_{i,j=1,i\neq j}^{n_1} D_{KL}(p_{ij}||p_{ij}(\vs_i,\vc_j,\ssp)) \\
=&\min_{\substack{\|\ttt_{1}\|_{\infty}\leq \tau\log n_1 \\ \ssp\in B_n^c(\rho)}} \sum_{i,j=1,i\neq j}^{n_1}D_{KL}(p_{ij}||p_{ij}(\vs_i,\vc_j,\ssp)) \gtrsim \rho n_1(n_1-1) e^{-12\|\ttt_1^*\|_{\infty}} > 0,
\end{aligned}
$$ which implies that $\epsilon_n(\rho)\gtrsim e^{-12\|\ttt_1^*\|_{\infty}}$. Since $\epsilon_n(\rho)$ is continuous, we get $2c_1\log n_1/n_1 = \epsilon_n(\rho_n)\gtrsim \rho e^{-12\|\ttt_1^*\|_{\infty}}$, which leads that $\rho_n \lesssim e^{12\|\ttt_1^*\|_{\infty}}\log n_1/n_1$. 
Therefore, $|\widetilde\ssp - \ssp_{01}|\lesssim_P e^{12\|\ttt_1^*\|_{\infty}}\log n_1/n_1$.
\end{proof}

Based on Lemmas~\ref{lem:selection} and \ref{lem:consis ssp}, we have the following proposition specifying the convergence rate of $\widetilde\sssp$, which satisfy the requirement in Theorem~\ref{thm:initial}.

\begin{proposition}
Under the same conditions of Lemma~\ref{lem:consis ssp}, 
suppose $n_1^*\gtrsim n$. Then, we have $\|\widetilde\sssp-\sssp^*\|_{\infty}\lesssim_P e^{12\|\ttt^*\|_{\infty}}\log n/n$.
\end{proposition}

\bibliographystyle{apalike}
\bibliography{ref}

\end{document}